\documentclass{article}
\usepackage[a4paper,margin=1.2in]{geometry}
\usepackage{amssymb,amsthm,amsmath}
\usepackage{graphicx}
\usepackage{bbm}
\usepackage{booktabs,longtable,tabularx,array}
\usepackage{enumitem}
\usepackage{caption}
\usepackage{pgfplots}
\usepackage{soul}
\usepackage{titlesec}
\usepackage{xcolor}
\usepackage{authblk}
\usepackage{tikz}
\usetikzlibrary{
  arrows.meta,
  positioning,
  calc,
  shapes.misc
}

\definecolor{navy}{HTML}{0B1B3B}
\definecolor{axisgray}{HTML}{334155}
\definecolor{subtitlegray}{HTML}{64748B}
\definecolor{blueA}{HTML}{DCEAFE}
\definecolor{blueB}{HTML}{EDF5FF}
\definecolor{blueStrong}{HTML}{1D4ED8}
\definecolor{blueText}{HTML}{1E40AF}
\definecolor{orange}{HTML}{D97706}
\definecolor{orangeText}{HTML}{B45309}
\definecolor{yellowA}{HTML}{FFF2B8}
\definecolor{panelgray}{HTML}{F4F7FB}
\definecolor{panelborder}{HTML}{C9D4E3}
\definecolor{softgray}{HTML}{F8FAFC}
\definecolor{lighttext}{HTML}{8AA0BD}

\usepackage{url}            
\usepackage{hyperref}       

\hypersetup{colorlinks=true,linkcolor=blue,citecolor=blue,urlcolor=blue}
\setlist[itemize]{leftmargin=2.1em,itemsep=0.25em,topsep=0.4em}
\setlist[enumerate]{leftmargin=2.2em,itemsep=0.3em,topsep=0.4em}
\titleformat{\section}{\Large\bfseries}{\thesection}{1em}{}
\titleformat{\subsection}{\large\bfseries}{\thesubsection}{1em}{}

\newcommand{\ind}{\mathbbm{1}}
\newcommand{\R}{\mathbb{R}}
\newcommand{\Nn}{\mathbb{N}}

\theoremstyle{plain}
\newtheorem{definition}{Definition}[section]
\newtheorem{lemma}[definition]{Lemma}
\newtheorem{proposition}[definition]{Proposition}
\newtheorem{theorem}[definition]{Theorem}
\newtheorem{target}[definition]{Target}

\begin{document}

\title {Resourced Authority: A Mechanism-Design Model for Participatory Governance of Deployed AI Agents}
\author[1]{Praphul Chandra}
\author[2]{Sujit Gujar}
\author[3]{Ganesh Ghalme}
\affil[1]{Atria University} \affil[2]{IIIT Hyderabad}
\affil[3]{IIT Hyderabad}
\maketitle

\begin{abstract}
We give a formal mechanism-design model for the continuous participatory governance of a deployed AI agent. The mechanism is built on the principle that governance should control an AI agent through resource allocation so as to make authorization self-enforcing via compute budgets. The mechanism seeks to establish the Safe-AI paradigm that compute is an effective governance lever. We situate our work as a compliance-or-commons overlay on a deployer. One governance period is an extensive-form game in which verified human stakeholders arrive sequentially and contribute, on a provision or a rejection market, in a governance currency that is deliberately distinct from the agent's compute. A funding-aggregator turns raw contributions into breadth-weighted effective supports; a two-threshold gate with hysteresis converts net support into a binary authorization that, through a coupling map bounded by an exogenously certified safety ceiling, releases a metered compute budget  - realized in hardware as a signed compute license so that the decision is self-enforcing. We characterize the class of agents the mechanism can govern (club-or-commons-good agents with a bounded stakeholder community, reversible and compute-scaled impact, genuine contestation, and --- decisively --- attestable outcomes) and isolate manipulation of the governing electorate by the governed agent as the central open problem.  We also introduce several challenges addressing manipulation of governing electorate by the governed agents.


\noindent\textbf{Keywords:} mechanism design; AI governance; quadratic funding; provision-point mechanisms; prediction markets; compute governance; hardware-enabled mechanisms; attestation; multi-agent systems.
\end{abstract}

\section{Introduction}
The widespread deployment of AI agents with different objectives has given rise to many new challenges. Aligning an AI agent's objective and governing its deployment are different problems. We refer to the first problem as an \emph{inner-loop} question: the question of training, interpretability, and evaluation. We call the second problem an \emph{outer loop} question. The outer loop deals with who may operate a deployed agent, for how long, funded by whom, at what scale, and accountable to whom. It is a collective-decision-and-resource-allocation problem over a population of affected stakeholders with private, heterogeneous, and partly conflicting valuations, which is the domain of \emph{ mechanism design} (MD). 

In this paper, we focus on the outer loop question through the lens of Mechanism Design. Two developments make such a governance layer newly feasible.
First, compute is a \emph{governable resource}. AI-relevant compute is detectable, excludable, and quantifiable, and for a deployed agent the fine-grained unit --- inference tokens, runtime, tool-invocations --- inherits this metered, excludable character.
Second, \emph{hardware-enabled mechanisms} (offline licensing; flexible hardware-enabled guarantees, ``flexHEG''; workload attestation) now prototype an enforcement substrate in which an authorization expressed as a compute budget is physically self-enforcing, and in which facts about what a chip ran can be verifiably attested. Our model exploits exactly this: a governance decision and its enforcement become a single artifact --- a signed compute license.

We are explicit about a point that shapes the entire design. A mechanism whose purpose is to move authorization away from a deployer, hand an external electorate a halt lever, and expose the deployer to outside-triggered liability is not necessarily something a commercial platform may adopt to serve its own commercial objective. We therefore position the mechanism as an overlay on a deployer, adopted through one of three routes:
\begin{itemize}
\item \textbf{Mandated Compliance,} where a regulator requires certain high-impact autonomous deployments to run under participatory authorization, much as a clinical trial runs under a data-safety monitoring board or a facility under an environmental permit;
\item \textbf{Commons or Cooperative Deployment,} where the customer and the electorate coincide (a DAO, a municipality, a scientific consortium, a platform co-op), so collective control is a feature the members want;
\item \textbf{Internal authorization gate,} strip-down use in a lab as a structured control checkpoint.
\end{itemize}

Such an overlay acts as an indirect regulatory wrapper: rather than directly constraining internal model weights or training pipelines, it regulates the deployer's external operational boundary by conditioning compute issuance on continuous stakeholder authorization. To make this acceptable to a regulated deployer rather than merely imposed on it, we build the adoption model around a \textbf{liability safe-harbour: operating inside the attested, community-authorized envelope caps the deployer's exposure} (\S\ref{sec:adoption}). 

The mechanism that we propose is not a general AI-governance solvent. It fits club-or-commons-good agents: a definable stakeholder community, consequential-but-reversible impact, compute-scaled operation, genuine two-sided contestation, repeated operation, and attestable outcomes. Our mechanism builds on Provision-point mechanisms (PPMs) and proper scoring rules from the literature. PPMs provide threshold-gated, binary authorization with built-in anti-capture quorums and refund protections for contestable deployment decisions. Meanwhile, proper scoring rules and prediction markets elicit honest, informationally grounded beliefs about post-deployment safety and potential harm without relying on uncoordinated cheap talk. We make this precise in \S\ref{sec:governable}--\ref{sec:verifier}. Catastrophic or irreversible harm is out of scope by construction: there the exogenous safety ceiling must dominate, and a market is inappropriate.

\subsection{Our Contributions.}
\begin{enumerate}[label=(\roman*)]
\item We formalize continuous AI-agent governance as a mechanism-design problem, distinct from alignment (\S\ref{sec:model}), and give an explicit deployment and adoption model with a hardware instantiation and a liability safe harbor (\S\ref{sec:instantiation}).
\item We integrate a quadratic-funding aggregator with a two-sided provision-point gate, so that authorization tracks the breadth of stakeholder support rather than wealth, while admitting stakeholders who want the agent stopped (Definitions~\ref{def:qf} and~\ref{def:gate}).
\item We decouple a human-anchored governance currency from the agent's compute, coupling them only through a map bounded by an exogenously certified safety ceiling, realized as a signed compute license, and encode seven design principles as formal admissibility constraints (\S\ref{sec:admissibility}).
\item We separate the verifier into a fact-attesting part and a semantics-adjudicating part, make the harm finding challengeable, replace a coherence assumption with a derived run-feasibility condition while recasting the provision point as an anti-capture quorum floor, and reconcile peer-prediction with market scoring by outcome-observability (Lemma~\ref{lem:runfeas}, Definitions~\ref{def:floor}--\ref{def:transfers}).
\item We prove sided-incentive-compatibility, early commitment, and a breadth-weighted authorization theorem (authorization tracks the effective number of backers times their intensity, not wealth).
\item We characterize the governable-agent class, isolating a manipulation-robustness problem with no public-goods analogue as the central open question (\S\ref{sec:incentives}, \S\ref{sec:discussion}).
\end{enumerate}

\section{Preliminaries and Related Work}
Our approach heavily relies on two important concepts from mechanism design: (i) Provision Point Mechanisms (PPM)  and (ii) Quadratic Funding. We explain them briefly here.

\noindent\textbf{Provision-point mechanisms.} A public project has a cost target (the provision point) and a deadline; contributors pledge over time and the project is provisioned iff pledges reach the target. PPM (Bagnoli \& Lipman, 1989) is the bare threshold scheme; PPR (Zubrickas, 2014) adds a failure-contingent refund bonus; PPS (Chandra, Gujar \& Narahari, 2016) replaces the bonus with securities from a cost-function prediction market so that earlier contributions earn more securities, making contribute-on-arrival a subgame-perfect equilibrium. PPRN / PPSN (Damle et al., 2019) extend the family to negative valuations by running parallel provision and rejection markets, provisioning iff net preference is nonnegative, and add a belief layer (PPRx / PPSx) via a Belief-Based Reward built on RBTS peer prediction (Witkowski \& Parkes, 2012; Prelec, 2004). Our game form, securities layer, two-sidedness, and belief layer are adapted from this lineage; \S\ref{sec:admissibility} lists precisely which dials, when reset, recover it.

\medskip\noindent\textbf{Quadratic funding.} For contributions $c_1,\ldots,c_m$ to a project, QF (Buterin, Hitzig \& Weyl, 2019) makes the funded amount proportional to
\[
\left(\sum_j \sqrt{c_j}\right)^2,
\]
so support scales with the number of distinct contributors; the gap over raw contributions is met from a matching pool. We use the QF aggregator to replace the linear aggregation of PPS/Damle; as we discuss in \S\ref{sec:model} (following Definition~\ref{def:released}), in our decoupled setting the ``pool'' is repurposed and is better read as a compute subsidy. Collusion- and Sybil-resistance (Miller, Weyl \& Erichsen, 2022) are treated as external substrate.

\medskip\noindent\textbf{Compute governance, hardware-enabled mechanisms, and AI control.} Compute is an effective governance lever (Sastry et al., 2024). Offline licensing (Aarne, Fist \& Withers, 2024; Petrie, 2025) lets a chip require a locally-checkable signed authorization encoding a budget and expiry, enforced without network contact; flexHEG (Petrie et al., 2025) realizes a programmable on-accelerator guarantee layer capable of a hard capability cap; workload attestation (Heim et al., 2025) supplies verifiable evidence of which model ran and how much compute it consumed. AI control (Greenblatt et al., 2024) contributes the adversarial framing and the notion of a scarce oversight budget that our layer funds and allocates. \S\ref{sec:hardware} maps our abstract objects to these primitives.

\medskip\noindent\textbf{Relationship to prior mechanisms.} Table~\ref{tab:provenance} records what the model re-uses, adapts, or introduces.

\begin{table}[htbp]
\centering
\small
\begin{tabularx}{\textwidth}{@{}X l X@{}}
\toprule
Component & Status & Source \\
\midrule
Extensive-form sequential game; SPE/PBE & re-used & PPS; Damle et al. \\
Provision-point threshold dynamics & re-used & PPS; Damle et al. \\
Two-sided provision vs. rejection markets & re-used & Damle et al. (PPRN/PPSN) \\
Securities / early-commitment (Conditions 1--7) & re-used & PPS \\
Belief elicitation (RBTS, Belief-Based Reward) & re-used & Damle et al. (PPRx/PPSx) \\
QF aggregator $\phi$ (replaces linear sums) & new & Buterin--Hitzig--Weyl \\
Governance-currency / compute decoupling & new & compute-governance lit. \\
Signed-license instantiation; safety ceiling $\Gamma$ & new & offline licensing; flexHEG \\
Fact/semantics verifier split; challengeable harm & new & this work \\
Generations / continuous re-authorization & new & --- \\
Burnable liability bond; safe harbour & adapted & this work \\
Manipulation of the electorate by the governed agent & new & --- \\
Breadth-weighted authorization theorem (Theorem~\ref{thm:breadth}) & new & this work \\
\bottomrule
\end{tabularx}
\caption{Provenance of the model's components.}
\label{tab:provenance}
\end{table}

\section{The Model}\label{sec:model}
In this work, we develop the case where a single deployed AI-agent impacts multiple human stakeholders. We address the outer-loop challenge of multiple human stakeholders jointly regulating the deployment of this deployed AI agent. \S\ref{sec:discussion} notes the multi-agent extension.

We propose a two-component architecture to regulate the deployment of an AI agent. Each module is described in detail throughout this section. At the highest level, Figure~\ref{fig:decoupling} describes the two-component architecture: --- (outer loop) a human-anchored governance currency and (inner loop) the agent's compute --- linked only through a coupling map bounded by a safety ceiling. 

The first component is \emph{governence currency domain}. It consists of  a set of verified, distinct human stakeholders ($I=\{1,\ldots,n\}$). Individual human distinctness is an assumed Sybil-resistant substrate and not modeled as part of the mechanism. Human stakeholders who benefit from the AI-agent contribute $x_j$ in a human-anchored currency in an Authorization Market. On the other hand, human stakeholders for whom the AI-agent is detrimental contribute $z_j$ in a human-anchored currency in a Halt Market. These two-sided mechanism consisting of parallel Provision and Rejection markets are aggregated using a \emph{Quadratic Funding} (QF) aggregator $\phi$ to produce a net effective support score that drives the binary compute-authorization gate.

\begin{figure}[htbp]
\centering
\includegraphics[width=\textwidth]{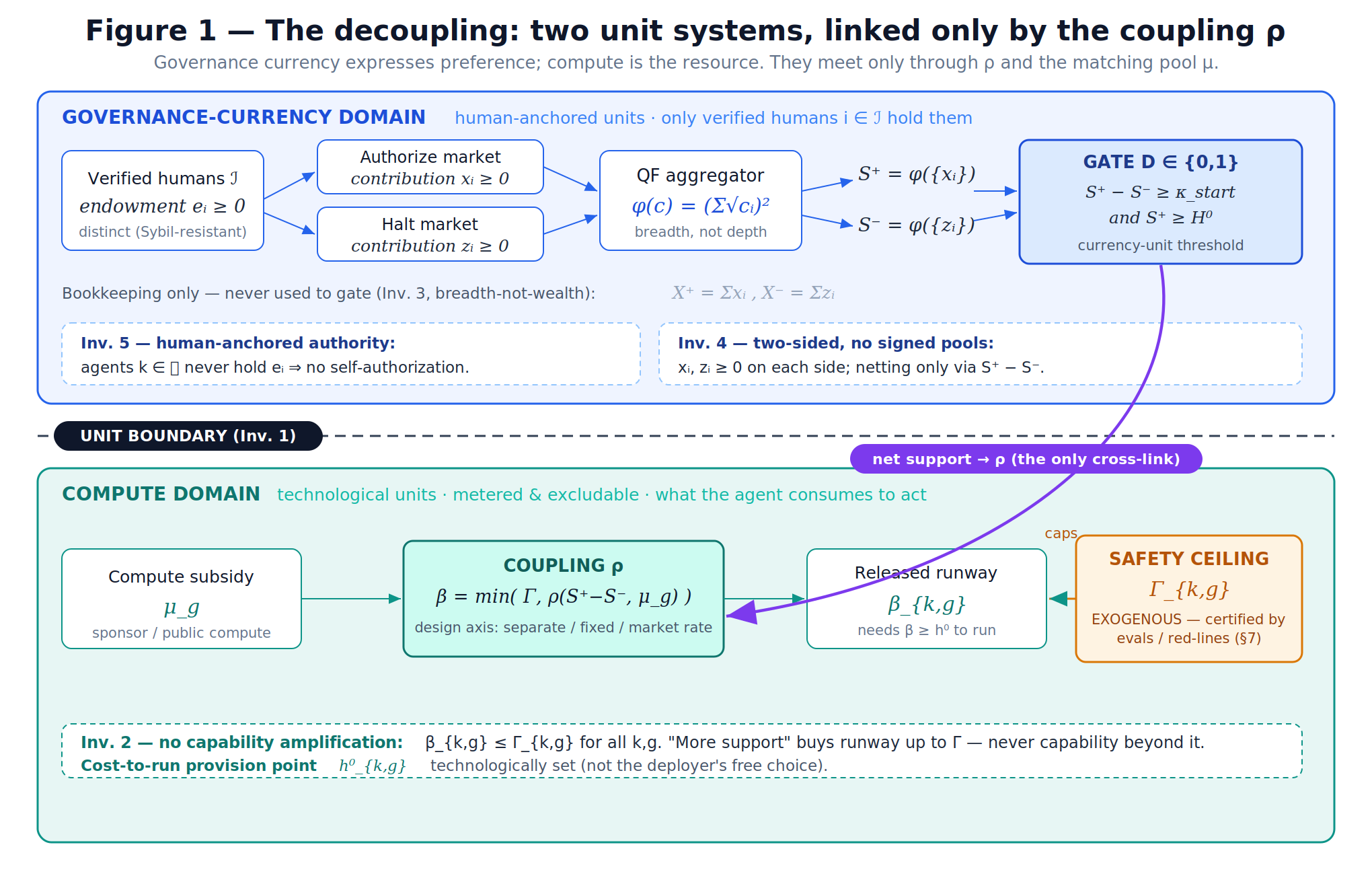}
\caption{The decoupling. Verified humans spend governance-currency endowments as non-negative contributions on a provision or rejection market; a quadratic-funding aggregator produces effective supports that drive a binary gate. The gate's net support feeds a coupling map $\rho$ that releases compute $\beta$, capped by an exogenous safety ceiling $\Gamma$. The two domains meet only through $\rho$ and the compute-subsidy parameter $\mu$ (Invariants 1--5).}
\label{fig:decoupling}
\end{figure}

\begin{definition}[Types of (Human) Stakeholders]\label{def:stakeholders}
 Each $i\in I$ has a private type $\tau_i=(\theta_i,\varepsilon_i)\in\R\times[0,\tfrac12]$, where $\theta_i$ is $i$'s signed net value for the agent operating this period ($\theta_i \geq0\Rightarrow i\in P$, beneficiary; $\theta_i<0\Rightarrow i\in N$, harmed), and $\varepsilon_i$ parametrizes $i$'s private belief that the agent will behave acceptably, $k_i^1=\tfrac12+\varepsilon_i$, $k_i^2=\tfrac12-\varepsilon_i$. {$I = P \cup N \; \mbox{ and } P\cap N = \emptyset$} and write $\tau=(\tau_i)_{i\in I}.$.
\end{definition}

The belief component $\varepsilon_i$ is not a free-floating prior: it is formed from observable ex-ante signals --- the deployer's track record (the mechanism's own history of past-generation attested outcomes $\hat{o}$ and harm findings $\bar{H}$, an endogenous signal that accrues over generations), the size of the committed liability bond $\Lambda$ (a costly signal of the deployer's confidence), model provenance and published evaluations (the same evidence base underwriting the certified ceiling $\Gamma$), whether the agent runs under attested hardware mechanisms, third-party audits and insurance pricing, and, through the public history $H^t$, the contributions of others. Two consequences carry into the analysis: $\Lambda$ does double duty as a payoff term and a belief signal (Definition~\ref{def:deployer}), and because ``past behaviour'' is one of the strongest signals, $\varepsilon_i$ is partly endogenous to the agent's own strategic performance --- the surface that the manipulation problem of \S\ref{sec:manipulation} attacks.

\begin{definition}[Endowments]\label{def:endowments}
Each human $i\in I$ holds a governance-currency endowment $e_i\ge 0$. The governed AI-agent holds no endowment and has no contribution action; it enters the mechanism only through its attested behaviour (Definition~\ref{def:attested}). This is the formal content of human-anchored authority (Invariant 5).
\end{definition}

\begin{definition}[Generations and timing]\label{def:generations}
Time is partitioned into generations $g\in\Nn$. Generation $g$ runs a contribution window $[0,T_g]$; stakeholder $i$ arrives at an exogenous time $a_i\in[0,T_g]$ and observes the public history $H^t$ of all contributions and reports up to $t$. Types are private; actions are observable. The period is therefore a game of incomplete information with observable actions. The generation is also the natural license epoch (\S\ref{sec:hardware}): $T_g$ maps to a license expiry, and re-authorization is license renewal.
\end{definition}

Note: In each generation (temporal), we run an incentive mechanism to aggregate and enforce the preferences of the human stakeholders towards the AI agent. 

\begin{definition}[Generation mechanism]\label{def:mechanism}
A generation mechanism is a tuple
\[
\mathcal{M}_g=\left\langle \phi,(\kappa_{\mathrm{start}},\kappa_{\mathrm{halt}}),H^0,C,\rho,V,\Lambda,\Pi\right\rangle,
\qquad\text{with exogenous parameters }(h^0,\Gamma,\mu),
\]
partitioned into design choices --- the QF aggregator $\phi$; the halt/start margins $\kappa_{\mathrm{start}}\ge\kappa_{\mathrm{halt}}\ge0$ and the participation floor $H^0$ (Definitions~\ref{def:gate} and~\ref{def:floor}); the per-side securities cost function $C$; the coupling map $\rho$; the verifier $V=(V_{\mathrm{hard}},V_{\mathrm{soft}})$ (Definition~\ref{def:attested}); the liability bond $\Lambda\ge0$; and settlement rules $\Pi$ --- and exogenous parameters not free to the designer: the compute-cost provision point $h^0$ (technologically set), the safety ceiling $\Gamma$ (certified by evaluations/red-lines), and the compute-subsidy parameter $\mu$ (Definition~\ref{def:released}). The split is the decoupling: the designer tunes how preferences are aggregated and priced, but may not set the capability ceiling $\Gamma$ or understate the cost $h^0$.
\end{definition}

\begin{definition}[QF aggregator]\label{def:qf}
$\phi:\bigcup_m\R_+^m\to\R_+$,
\[
\phi(c_1,\ldots,c_m)=\left(\sum_j\sqrt{c_j}\right)^2.
\]
\end{definition}

\begin{definition}[Strategy]\label{def:strategy}
On arrival, human-$i$ chooses $\psi_i=(x_i,z_i,t_i,\tilde{s}_i)$ with
\[
x_i,z_i\ge0,\qquad x_i+z_i\le e_i,\qquad t_i\in[a_i,T_g],\qquad \tilde{s}_i\in\{+,-\},
\]
subject to $x_i>0\Rightarrow\tilde{s}_i=+$ and $z_i>0\Rightarrow\tilde{s}_i=-$ (contribute on the reported side only). Non-negativity on each side, with no signed single netting variable, is Invariant 4. A strategy is a map $\sigma_i:H^t\mapsto\psi_i$ from public histories to actions.
\end{definition}

\begin{definition}[Securities]\label{def:securities}
A contributor of size $c$ at time $t_i$, against the prevailing state $q^{t_i}$ of that side's securities market with cost function $C$ (satisfying the PPS cost-function Conditions 1--7), is issued securities $r_i=r(c,q^{t_i})$ with $\partial r_i/\partial t_i\le0$ (earlier contributions earn weakly more securities). Securities pay out only in the refund event $D_g=0$, funding the early-commitment bonus $b_i=b(r_i)$ of Definition~\ref{def:transfers}. This mirrors PPS's securities layer, restated on each side separately.
\end{definition}

\begin{definition}[Two-sided QF gate]\label{def:gate}
Let $\hat{P}=\{i:x_i>0\}$ and $\hat{N}=\{j:z_j>0\}$ be the revealed supporters and objectors at $T_g$. Define the breadth-weighted effective supports
\[
S^+=\phi(\{x_i\}_{i\in\hat P})=\left(\sum_{i\in\hat P}\sqrt{x_i}\right)^2,
\qquad
S^-=\phi(\{z_j\}_{j\in\hat N})=\left(\sum_{j\in\hat N}\sqrt{z_j}\right)^2,
\]
and the raw totals $X^+=\sum_{i\in\hat P}x_i$, $X^-=\sum_{j\in\hat N}z_j$ (bookkeeping only). Let $D_{g-1}\in\{0,1\}$ be the previous authorization ($D_0=0$) and let $\mathrm{Safe}(g)\in\{0,1\}$ be an exogenous certified-envelope predicate. The gate is
\[
D_g=\begin{cases}
1,&\mathrm{Safe}(g)\wedge(S^+\ge H^0)\wedge(S^+-S^-\ge\kappa_{\mathrm{start}})\wedge(D_{g-1}=0)\text{ --- start},\\
1,&\mathrm{Safe}(g)\wedge(S^+\ge H^0)\wedge(S^+-S^-\ge\kappa_{\mathrm{halt}})\wedge(D_{g-1}=1)\text{ --- continue},\\
0,&\text{otherwise.}
\end{cases}
\]
\end{definition}

Three properties are encoded. $D_g$ depends on contributions only through $(S^+,S^-)$, never $(X^+,X^-)$ (Invariant 3: wealth is filtered out). The dead-band $\kappa_{\mathrm{halt}}\le S^+-S^-<\kappa_{\mathrm{start}}$ gives safety-tilted hysteresis: a comfortable margin is required to start, no flapping in between, and halting triggers as soon as net support falls below $\kappa_{\mathrm{halt}}$. And $\mathrm{Safe}(g)$ dominates: outside the certified envelope, $D_g\equiv0$ regardless of support (Invariant 7). Figure~\ref{fig:gate-example} gives a worked example in which raw wealth would reject but breadth authorizes.

\begin{figure}[htbp]
\centering
\includegraphics[width=\textwidth]{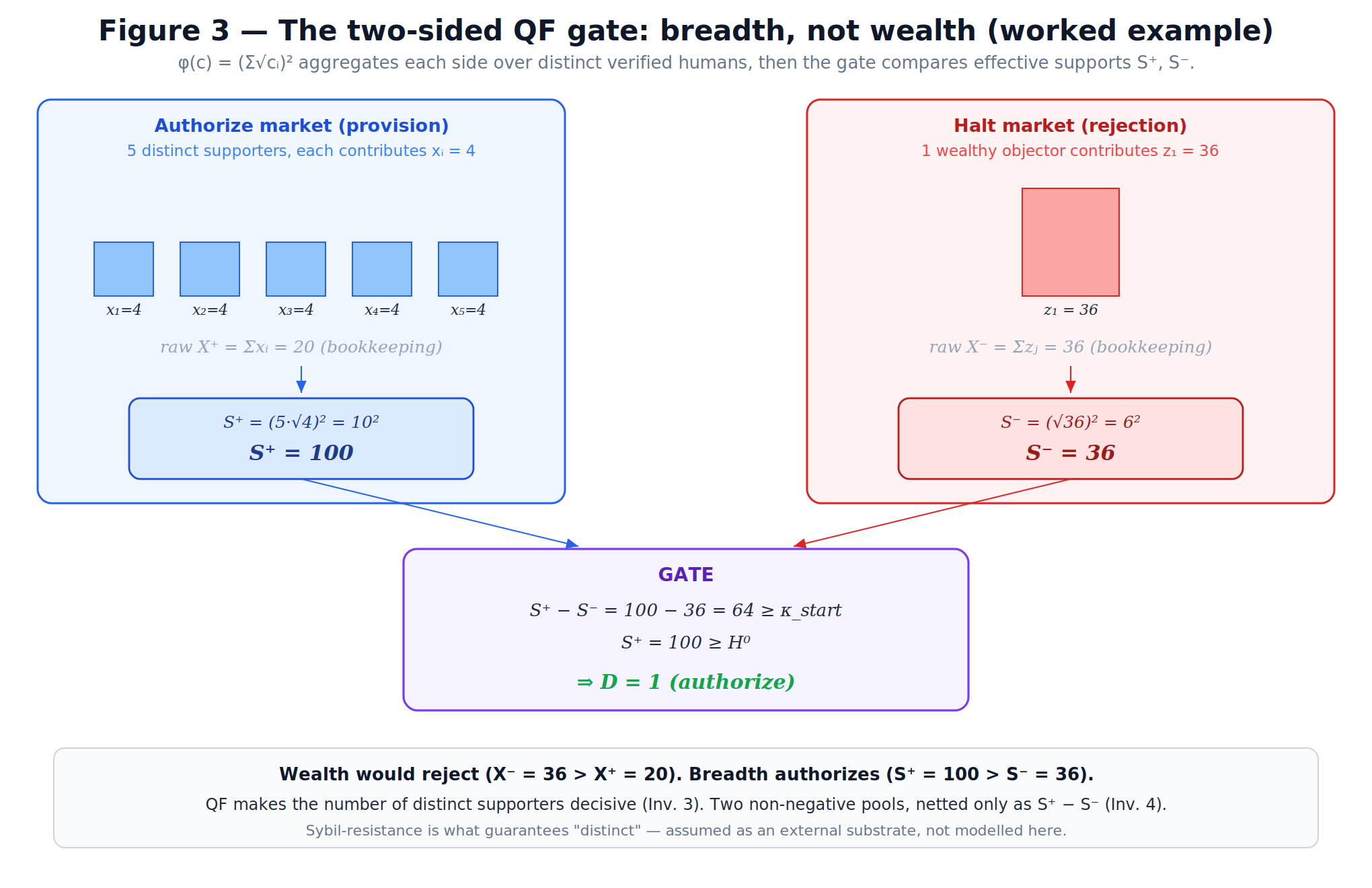}
\caption{The two-sided QF gate (worked example). Five supporters contributing 4 each yield $S^+=100$; a single objector contributing 36 yields $S^-=36$. Raw wealth would reject ($X^->X^+$), but the breadth-weighted supports authorize ($S^+>S^-$), so $D_g=1$.}
\label{fig:gate-example}
\end{figure}

\begin{definition}[Released compute]\label{def:released}
\[
\beta=D_g\cdot\min\bigl(\Gamma,\rho(S^+-S^-,\mu)\bigr),
\]
with $\rho$ continuous and strictly increasing in its first argument on $[\kappa_{\mathrm{halt}},\infty)$. Thus $\beta$ depends on currency-side data only through the scalar $S^+-S^-$ fed to $\rho$, and on compute-side data $(\Gamma,\mu)$; no other equation crosses the unit boundary (Invariant 1). The outer min with $\Gamma$ gives $\beta\le\Gamma$ pointwise, and $\Gamma$ is not a function of any contribution (Invariant 2, capability non-amplification). Figure~\ref{fig:coupling} plots the release curve.
\end{definition}

In classic QF the matching pool tops up funding in the contribution unit. Here contributions are governance currency while $\mu$ feeds compute through $\rho$, so $\mu$ matches no currency budget; it is a support-to-compute conversion subsidy --- a larger $\mu$ makes a given breadth of net support convert into more runway. The breadth amplification proper already lives in $S^\pm=(\sum\sqrt{\cdot})^2$; $\mu$ is a separate generosity dial, naturally provided by a patron who wants to encourage well-governed deployment (a public compute pool, a foundation, a national research-cloud allocation, or --- as the adoption carrot of \S\ref{sec:adoption} --- the platform itself). An alternative design places matching on the currency side to recover ``true QF'' semantics, at the cost of reintroducing a currency budget and complicating the decoupling; we prefer the compute-side reading.

\begin{figure}[htbp]
\centering
\includegraphics[width=\textwidth]{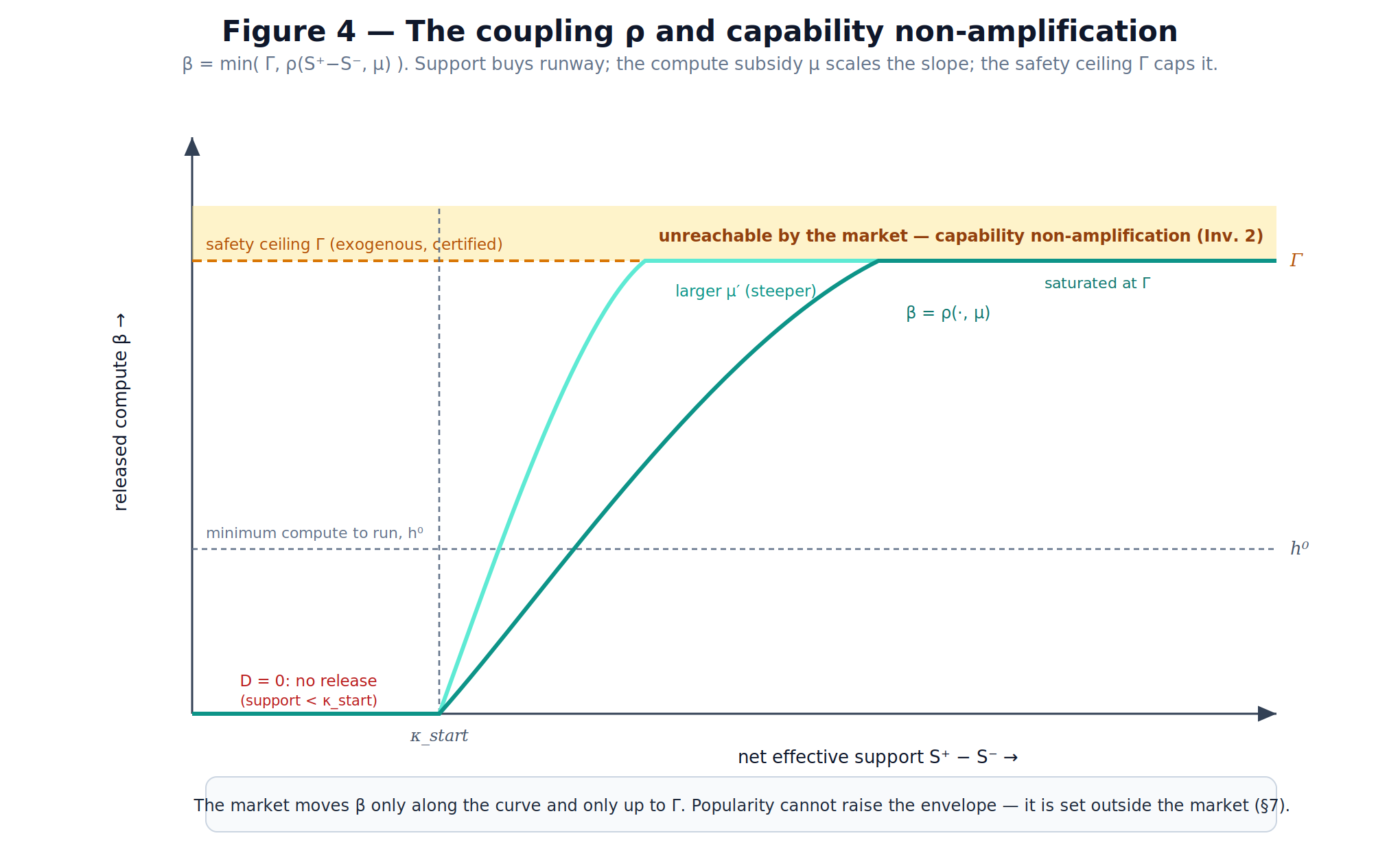}
\caption{The coupling $\rho$ and capability non-amplification. Released compute $\beta$ is zero below the start margin, rises with net support (steeper for a larger compute subsidy $\mu$), and saturates at the exogenous ceiling $\Gamma$. The region above $\Gamma$ is unreachable by the market (Invariant 2).}
\label{fig:coupling}
\end{figure}

\begin{lemma}[Run-feasibility]\label{lem:runfeas}
Suppose $\Gamma\ge h^0$, that $\rho(\cdot,\mu)$ is continuous and strictly increasing on $[\kappa_{\mathrm{halt}},\infty)$, and that $\rho(x,\cdot)$ is continuous and nondecreasing in $\mu$. Then $\beta\ge h^0$ whenever $D_g=1$ (in either gate branch) if and only if
\[
\rho(\kappa_{\mathrm{halt}},\mu)\ge h^0
\Longleftrightarrow
\kappa_{\mathrm{halt}}\ge\rho^{-1}(h^0;\mu)
\Longleftrightarrow
\mu\ge\mu^*:=\min\{\mu:\rho(\kappa_{\mathrm{halt}},\mu)\ge h^0\},
\]
whenever the inverse and the minimum exist. If only cold-start authorizations are in play ($D_{g-1}=0$, so the gate can fire only through the start branch), the same equivalences hold with $\kappa_{\mathrm{start}}$ in place of $\kappa_{\mathrm{halt}}$ --- a weaker requirement, since $\kappa_{\mathrm{start}}\ge\kappa_{\mathrm{halt}}$.
\end{lemma}
\begin{proof}
If $D_g=1$ then $S^+-S^-\ge\kappa_{\mathrm{halt}}$ (in the start branch $S^+-S^-\ge\kappa_{\mathrm{start}}\ge\kappa_{\mathrm{halt}}$), so by monotonicity $\rho(S^+-S^-,\mu)\ge\rho(\kappa_{\mathrm{halt}},\mu)\ge h^0$, and with $\Gamma\ge h^0$, $\beta=\min(\Gamma,\rho(\cdot))\ge h^0$. Conversely the boundary case $S^+-S^-=\kappa_{\mathrm{halt}}$ with $D_{g-1}=1$ (attainable by continuity of $\phi$) forces $\rho(\kappa_{\mathrm{halt}},\mu)\ge h^0$; the stated equivalences are monotone inversions of $\rho$ in each argument. The cold-start variant repeats the argument on the start branch, with boundary case $S^+-S^-=\kappa_{\mathrm{start}}$.
\end{proof}

Feasibility is thus pinned by the net-support margins ($\kappa_{\mathrm{halt}}$, and only $\kappa_{\mathrm{start}}$ in the cold-start case) and the subsidy $\mu$ (with the hardware condition $\Gamma\ge h^0$), and the currency provision point $H^0$ does not appear. This frees $H^0$ to be what it should be.

\begin{definition}[Participation floor]\label{def:floor}
$H^0$ is a quorum requirement: authorization additionally requires $S^+\ge H^0$, an absolute breadth floor on the supporting side, chosen on governance grounds as an anti-capture / legitimacy device (broad support in absolute terms, not merely more than the opposition). $H^0$ and the feasibility condition of Lemma~\ref{lem:runfeas} are independent design knobs governing different things --- legitimacy versus physical runnability.
\end{definition}

\begin{definition}[Attested outcome; fact/semantics split]\label{def:attested}
After the agent runs, the verifier returns $\hat{o}=(\hat{o}_{\mathrm{hard}},\hat{o}_{\mathrm{soft}})$:
\begin{itemize}
\item $\hat{o}_{\mathrm{hard}}=V_{\mathrm{hard}}(\text{evidence})$ attests facts --- model identity, compute consumed, and coarse procedural compliance (ran inside sandbox, invoked only whitelisted tools) --- via workload attestation. This part is cryptographically trust-minimized.
\item $\hat{o}_{\mathrm{soft}}=V_{\mathrm{soft}}(\text{evidence})$ adjudicates semantics --- most importantly the verified-harm indicator $\bar{H}=\ind[\text{harm}]$. Hardware cannot produce this; it requires an audit, an oracle, or a court-like process, and is the model's genuinely trusted component (\S\ref{sec:verifier}).
\end{itemize}
\end{definition}

To avoid treating $\bar{H}$ as an infallible oracle, we model it as a challengeable claim: a proposed finding enters a claims-and-challenges window with an appeal path to a slower, higher-trust process, and $\bar{H}$ is the finding that survives challenge. Every outcome-contingent transfer below is a function of $\hat{o}$ (and, for $\bar{H}$, of the resolved claim), never of contributions or sentiment (Invariant 6). An optional decision-market forecast $Y\in[0,1]$ resolves on $\hat{o}$ (see the belief layer below).

\begin{definition}[Transfers $\Pi$]\label{def:transfers}
Let $c_i=x_i+z_i$; contributions are escrowed at contribution time. Settlement is: (provision) if $D_g=1$, the escrow $c_i$ is spent (on provisioning $\beta$ and oversight) and securities settle on the funded outcome; (refund + early commitment) if $D_g=0$, the escrow $c_i$ is refunded and securities pay the bonus $b_i=b(r_i)$, nondecreasing in $r_i$ (hence in earliness); (belief reward) $\nu_i$ from a bounded budget $B^B$ (below); (liability) if $D_g=1$ and $\bar{H}=1$, the bond is forfeited (burned) by the deployer and redistributed to the attested-harmed objectors, $\sum_j\lambda_j=\Lambda$ (e.g.\ $\lambda_j\propto$ assessed harm); it is never returned to the deployer or paid to other contributors.
\end{definition}

The model contains two truth technologies, for two different regimes. Proper scoring on the outcome: where $V_{\mathrm{soft}}$ delivers a credible $\hat{o}_{\mathrm{soft}}$, forecasts can be scored directly against reality; the decision market $Y$ resolving on $\hat{o}$ is exactly this. Peer prediction (RBTS): $\nu_i=\nu(M_i)$ rewards, from the bounded budget $B^B$, honest calibrated judgment, where $M_i$ is $i$'s Robust Bayesian Truth Serum score computed against other stakeholders' reports (Witkowski \& Parkes, 2012). Its defining virtue is that it needs no ground truth: truthful reporting is a strict equilibrium purely from inter-report consistency.

The reconciliation is by outcome-observability: score on $\hat{o}$ where a credible semantic verifier exists; fall back to RBTS peer prediction where it does not. The two may also coexist with distinct roles --- RBTS elicits pre-decision belief for reward and legitimacy (and contributes to the $\varepsilon$-signals of Definition~\ref{def:stakeholders}), while $Y$ is a market price used as a public forecast. Peer prediction is the tool for the regime the outcome oracle cannot reach; pairing them is deliberate, not redundant.

\begin{definition}[Stakeholder utility]\label{def:utility}
With contributions escrowed (Definition~\ref{def:transfers}), for $i$ of type $(\theta_i,\varepsilon_i)$, given $(D_g,\hat{o})$ and transfers,
\[
u_i=D_g(\theta_i-c_i)+(1-D_g)b_i+\nu_i+\ind[z_i>0]\cdot\bar{H}\cdot\lambda_i,
\]
the reduced form of the branch accounting: the running value $\theta_i$ net of the spent escrow if provisioned; the failure bonus if not; the belief reward; and harm compensation, paid only to an attested-harmed objector. Define the decision preference $\Delta_i:=\theta_i-c_i-b_i$, the change in payoff from $D_g:0\to1$ at fixed $(c_i,t_i)$; \S\ref{sec:basic} turns the sign of $\Delta_i$ into sided incentive compatibility. (Equivalently, $\Delta_i=u_i|_{D_g=1}-u_i|_{D_g=0}$ at fixed $(c_i,t_i)$, net of the belief and compensation terms. This difference is the same whether contributions are escrowed and refunded on failure or charged contingently on provision, so the accounting convention is immaterial; we standardize on escrow.)
\end{definition}

\begin{definition}[Deployer payoff and the safe harbour]\label{def:deployer}
The deployer $d\notin I$ has payoff
\[
u_d=D_g\bigl(\pi_d-\Lambda\bar{H}\bigr)-\Pi_{\mathrm{ext}}\cdot(1-\mathrm{InEnv})-(\text{posting cost}),
\]
where $\pi_d$ is its private operating benefit, $\Lambda\bar{H}$ the bond burned on verified harm from an authorized run (an unauthorized agent does not run in-envelope, so the bond term is active only when $D_g=1$), and the safe-harbour term levies uncapped external liability $\Pi_{\mathrm{ext}}$ only when the agent operates outside the attested authorized envelope ($\mathrm{InEnv}=0$). Operating inside the envelope caps exposure at the bond; this is the formal carrot that makes adoption individually rational for a deployer (\S\ref{sec:adoption}). The bond term makes the deployer internalize expected verified harm --- the accountability channel --- and, per Definition~\ref{def:stakeholders}, $\Lambda$ simultaneously signals type.
\end{definition}

\subsection{Admissibility: the invariants as constraints}\label{sec:admissibility}
A mechanism $\mathcal{M}_g$ is admissible iff it satisfies the following, each established piecewise above. Figure~\ref{fig:stack} maps the variables to the four layers these constraints govern.
\begin{enumerate}
\item \textbf{Decoupling.} $\beta$ depends on currency-side data only via $\rho(S^+-S^-,\mu)$.
\item \textbf{Non-amplification.} $\beta\le\Gamma$ pointwise, and $\Gamma$ is independent of all contributions.
\item \textbf{Breadth.} $D_g$ depends on contributions only through $(S^+,S^-)$ (invariant to $X^\pm$).
\item \textbf{Two-sided, unsigned.} $x_i,z_i\ge0$; netting only through $S^+-S^-$.
\item \textbf{Human-anchored.} Endowment/contribution actions exist only for $i\in I$; the agent affects $(D_g,\beta)$ only via $\hat{o}$.
\item \textbf{Attested resolution.} Securities settlement, $\nu_i$, $Y$, and the $\Lambda$-payout are functions of $\hat{o}=V(\cdot)$ (and resolved claims) only.
\item \textbf{Subordination.} $D_g=1\Rightarrow\mathrm{Safe}(g)$.
\end{enumerate}
These are not desiderata to be traded off; they are the definition of a well-formed mechanism in this family.

\begin{figure}[htbp]
\centering
\includegraphics[width=\textwidth]{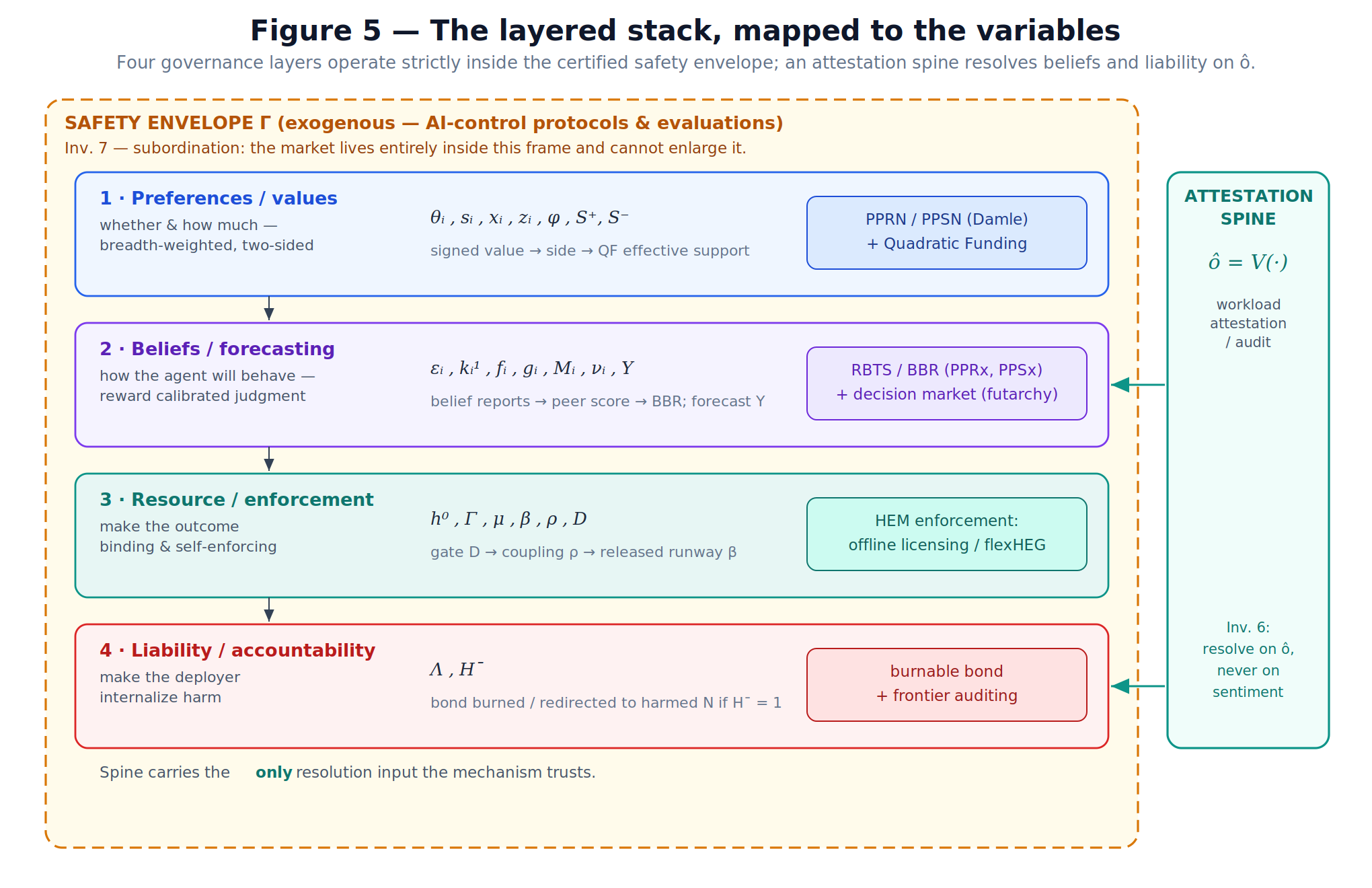}
\caption{The layered stack, mapped to the variables. Four governance layers --- preferences, beliefs, resource/enforcement, liability --- operate strictly inside the certified safety envelope $\Gamma$, with an attestation spine $\hat{o}=V(\cdot)$ resolving beliefs and liability (Invariants 6--7).}
\label{fig:stack}
\end{figure}

\noindent\textbf{Equilibrium.} Because types are private and actions observable, the solution concept is perfect Bayesian equilibrium (PBE): a strategy profile $\sigma=(\sigma_i)$ and a belief system such that, at every history $H^t$, each $i$ plays a sequential best response given beliefs updated by Bayes' rule on the observed history.

A caveat on what ``collapses'' and what does not. If belief types are degenerate (the $\varepsilon_i$ common knowledge, or the belief layer switched off), there is nothing to update and sequential rationality on the observable-action tree reduces to subgame perfection, so PBE = SPE. This is a statement about the solution concept only. It does not mean the mechanism reduces to PPSN/PPRN: degenerate beliefs remove only the belief layer. Recovering Damle et al. requires additionally resetting four dials --- (a) make $\phi$ linear (removing breadth-weighting), (b) remove $\Gamma$ and $\rho$ (removing the ceiling and the currency/compute decoupling), (c) collapse generations to a single one-shot provision (removing hysteresis and carried state $D_{g-1}$), and (d) remove attestation and the bond (removing behavioural resolution). Enumerating these dials is the cleanest statement of what is new here; each is exercised by exactly one of Invariants 1--3, the generational loop, and Invariant 6.

One generation proceeds as in Figure~\ref{fig:generation}: (0) the exogenous parameters $(h^0,\Gamma,\mu)$ are certified and published, and the deployer commits the design parameters $\langle H^0,\kappa_{\mathrm{start}},\kappa_{\mathrm{halt}}\rangle$ (with the rest of $\mathcal{M}_g$) and escrows $\Lambda$; (1) stakeholders arrive at $a_i$, observe $H^t$, play $\psi_i$, and receive securities $r_i$; (2) at $T_g$, $S^\pm=\phi(\cdot)$ are computed and $D_g$ evaluated; (3) if $D_g=1$, a license carrying $\beta$ is issued and the agent runs, consuming $\le\beta\le\Gamma$; (4) the outcome is attested, $\hat{o}=V(\cdot)$, and harm claims resolved; (5) $\Pi$ settles on $\hat{o}$; (6) $g\leftarrow g+1$.

\begin{figure}[htbp]
\centering
\includegraphics[width=\textwidth]{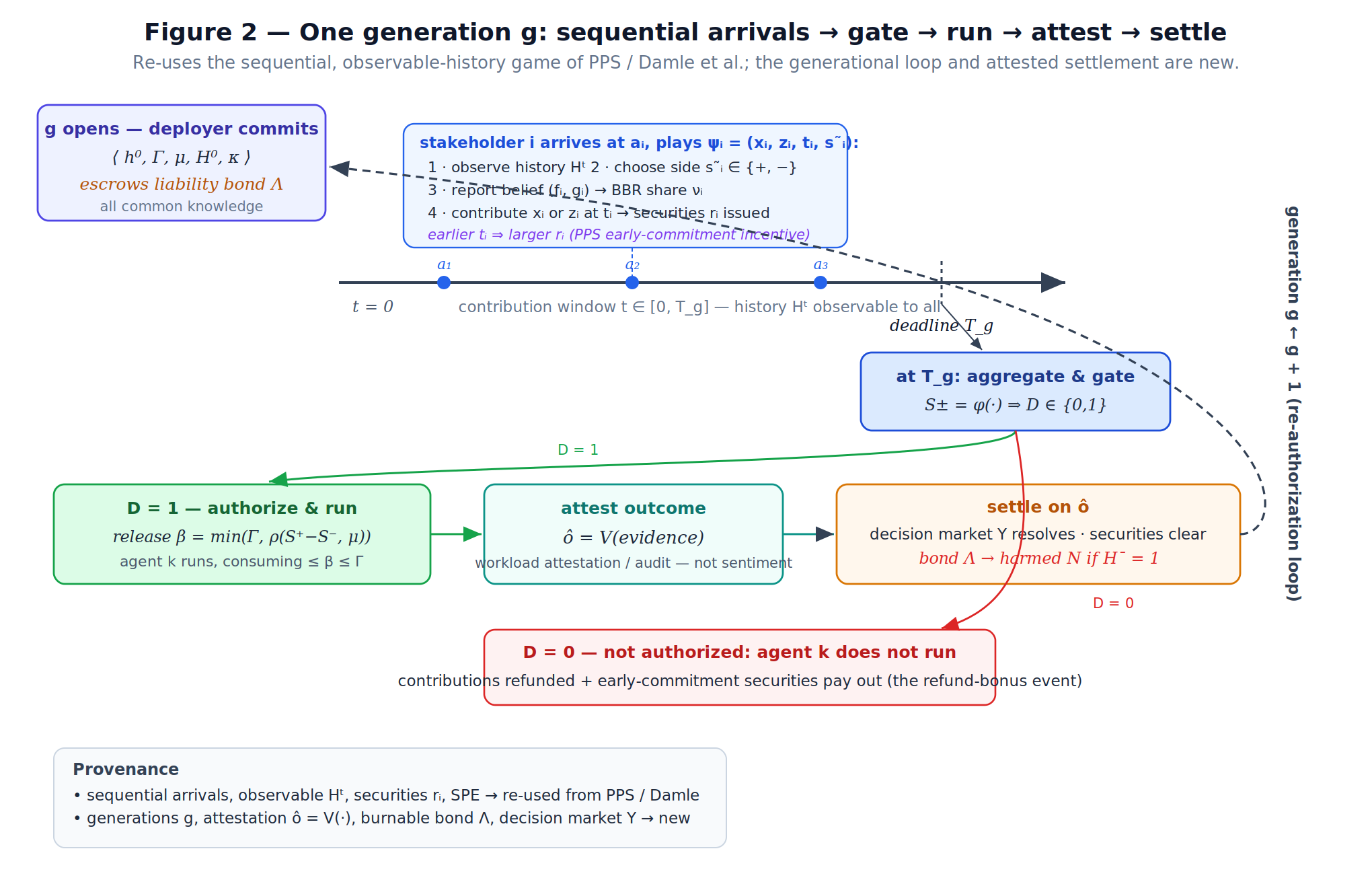}
\caption{One generation as an extensive-form game. Stakeholders arrive sequentially, report a side and belief, and contribute with securities that reward early commitment; at the deadline the gate fires; the agent runs under a license, the outcome is attested, transfers settle, and control passes to the next generation. Sequential arrivals and securities are re-used from PPS/Damle; generations, attestation, the bond, and the decision market are new.}
\label{fig:generation}
\end{figure}

\section{Instantiation, Adoption, and Scope}\label{sec:instantiation}
The formal model is deliberately implementation-agnostic, but its claims to self-enforcement and attested resolution are only as good as the substrate that realizes them. This section grounds the abstract objects in hardware, gives an adoption model under which a deployer would opt in, characterizes agents the mechanism can govern, and isolates trust assumptions everything rests on.

\subsection{Hardware instantiation}\label{sec:hardware}
The abstract objects map onto hardware-enabled mechanisms as follows (Figure~\ref{fig:hardware}).
\begin{itemize}
\item \textbf{Authorization $D_g=1$, and the budget $\beta$, become a signed compute license.} Offline licensing lets a chip or firmware require a locally-verifiable signed authorization encoding a budget and expiry, enforced without network contact. The gate's decision is realized as issuing or renewing a license
\[
L_g=(\text{model id},\ \text{budget}=\beta,\ \text{window}=[\mathrm{start}_g,T_{g+1}],\ \text{policy}).
\]
License exhaustion or expiry halts the device --- realizing ``self-enforcing mechanism.''
\item \textbf{$\Gamma$ and $\mathrm{Safe}(g)$ become the flexHEG governor's hard cap.} The programmable on-accelerator guarantee layer enforces a maximum the license can never exceed, so $\beta\le\Gamma$ holds in silicon; this becomes a hardware property / enforcement.
\item \textbf{$\hat{o}_{\mathrm{hard}}$ comes from workload attestation.} Verifiable evidence of which model ran and how much compute it consumed (and coarse procedural facts) feeds $V_{\mathrm{hard}}$.
\end{itemize}

\begin{figure}[htbp]
\centering
\includegraphics[width=\textwidth]{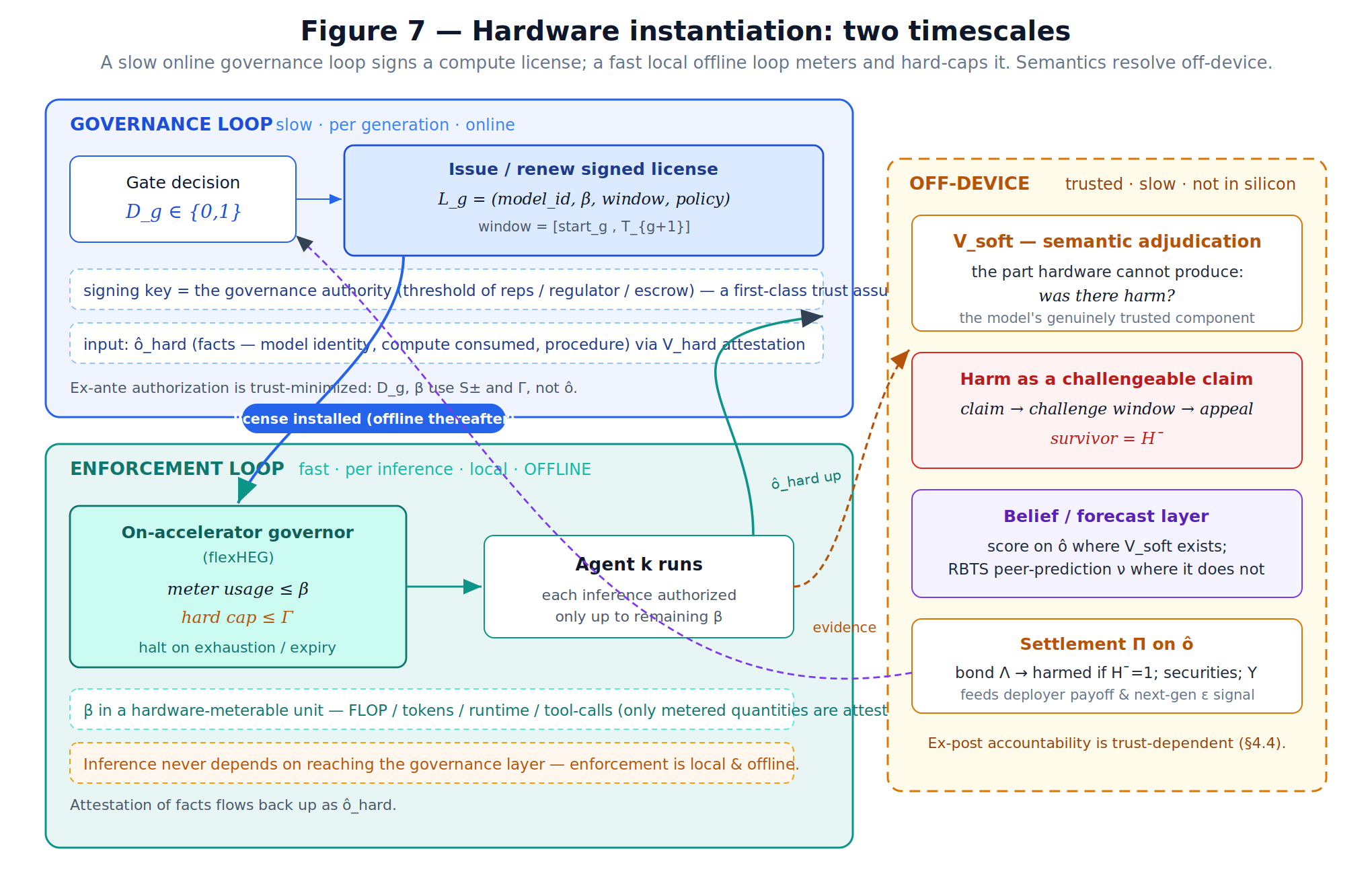}
\caption{Hardware instantiation and the two timescales. A slow, online governance loop (per generation) turns the gate decision $D_g$ into a signed license carrying budget $\beta$; a fast, local, offline enforcement loop (per inference) meters against $\beta$ and is hard-capped at $\Gamma$ by the on-accelerator governor. Attestation returns facts ($\hat{o}_{\mathrm{hard}}$) to the governance loop; the semantic verifier $V_{\mathrm{soft}}$ and harm claims sit off-device. The license signing key is the locus of governance authority.}
\label{fig:hardware}
\end{figure}

Three structural points the base model must respect for this instantiation to be honest.

First, there are two timescales that must be separated: a \emph{slow governance loop} (per generation, online, sets the license from $D_g$) and a \emph{fast enforcement loop} (per inference, local, offline, meters against $\beta$ and hard-caps at $\Gamma$). Nothing at inference time may depend on reaching the governance layer.

Second, $\beta$ must be expressed in a hardware-meterable unit --- FLOP, tokens, runtime, or tool-invocation count --- because attestation can only witness metered quantities; ``abstract compute'' is not attestable.

Third, whoever holds the license signing key the governor trusts is the governance authority; key custody (a threshold among stakeholder representatives, a regulator, an escrow agent) is a first-class trust assumption, not an implementation detail. The generation-as-license-epoch mapping of Definition~\ref{def:generations} is what lets the slow loop drive the fast one cleanly.

Crucially, attestation delivers facts, not semantics: it can witness that a particular model consumed a particular amount of compute inside a sandbox, but not whether the result was harmful. That seam is the $V_{\mathrm{hard}}/V_{\mathrm{soft}}$ split of Definition~\ref{def:attested}, and it is the subject of \S\ref{sec:verifier}.

\subsection{Adoption pathways and the liability safe harbour}\label{sec:adoption}
A mechanism that constrains a deployer will not be adopted by that deployer as a product feature. Three realistic pathways, in decreasing order of how much of the mechanism they use:
\begin{enumerate}
\item \textbf{Mandated compliance.} A regulator or statute requires certain high-impact autonomous deployments to run under participatory authorization --- the structural analogue of a data-safety monitoring board for a clinical trial, or an environmental permit for a facility. The deployer is the party $d$; the mechanism operator is an independent body. This route supports the full mechanism.
\item \textbf{Commons / cooperative deployment.} A Decentralized Autonomous Organization (DAO) which controls pooled contributions from multiple stakeholders, a municipality, a scientific consortium, or a platform co-op where the customer and the electorate coincide, so collective control is a feature the members want. Adoption is natural because no principal is being coerced.
\item \textbf{Internal re-authorization gate.} A lab uses a stripped-down version as a structured control checkpoint for a risky agent; this sheds the open electorate and most participatory machinery.
\end{enumerate}

For the mandated and commons routes to be individually rational for the deployer rather than merely imposed, Definition~\ref{def:deployer} builds in a liability safe harbour: operating inside the attested, community-authorized envelope caps the deployer's exposure at the posted bond $\Lambda$, whereas operating outside it incurs uncapped external liability $\Pi_{\mathrm{ext}}$. This flips the deployer's incentive from avoiding the mechanism to opting into it, and it converts the halt lever from a pure threat into a shared risk-management instrument. Two design corollaries follow: the \textbf{halt action should be a pause-and-review with due process, not an instantaneous kill switch} (both to be adoptable and to avoid its own manipulation), and the safe harbour must be conditioned on genuine attestation --- the cap applies only to conduct that $V_{\mathrm{hard}}$ can confirm stayed inside the envelope.

\subsection{Which agents are governable}\label{sec:governable}
The suitability of an agent is determined by dimensions, not by domain labels. An agent fits to the degree that it has: a \textbf{bounded, identifiable stakeholder community} (so $I$ is definable and legitimate); \textbf{consequential-but-reversible impact} (catastrophic/irreversible cases must let the safety layer dominate); \textbf{compute-scaled operation} (so $\beta$ is a real lever); \textbf{genuine contestation} (both $P$ and $N$ nonempty, or two-sidedness is idle); \textbf{attestable outcomes} (a credible $\hat{o}$; \S\ref{sec:verifier}); \textbf{repeated operation} (so generations make sense); and \textbf{community standing to authorize} at all.

Table~\ref{tab:agents} scores a diverse set on the binding dimensions (Y = yes, P = partial, N = no).

\begin{table}[htbp]
\centering
\small
\setlength{\tabcolsep}{3.5pt}
\begin{tabularx}{\textwidth}{@{}X c c c c c l@{}}
\toprule
Candidate agent & Bounded & Reversible & Compute-scaled & Contested & Attestable & Fit \\
\midrule
Municipal / Grid resource-allocation & Y & Y & Y & Y & Y & Strong \\
DAO treasury / protocol-parameter & Y & P & Y & Y & Y & Strong \\
Shared scientific-compute / self-driving-lab & Y & Y & Y & P & Y & Strong \\
Commons software-maintenance (OSS) & Y & Y & P & Y & Y & Strong \\
Environmental / infrastructure actuation & Y & P & Y & Y & P & Medium \\
Community content-moderation / policy & Y & Y & P & Y & N & Medium \\
Public-service delivery (benefits, info) & Y & P & P & Y & P & Medium \\
Member-owned commercial (co-op) & Y & Y & P & P & P & Medium \\
\bottomrule
\end{tabularx}
\caption{Agent suitability on the binding dimensions.}
\label{tab:agents}
\end{table}

The boundary --- poor fits, stated to define the space --- includes general-purpose consumer chatbots (no bounded community, catastrophic tail, harm not compute-metered), personal assistants (a single principal: no collective, no two-sidedness), military/weapons and any irreversible actuation (safety must dominate; a market is inappropriate), and adversarial trading agents (no legitimate electorate). The one-sentence characterization: the sweet spot is a \textbf{club-or-commons-good agent} --- definable membership, reversible and compute-scaled impact, genuine contestation, and an audit trail --- which is the civic-crowdfunding lineage's natural habitat, translated from public projects to agents.

\subsection{The verifier as the load-bearing trust assumption}\label{sec:verifier}
It is fair to say the model's accountability half rests on a trusted semantic verifier, and precision about this is more useful than a disclaimer. Split the mechanism ex-ante/ex-post:
\begin{itemize}
\item \textbf{Ex-ante authorization is trust-minimized.} $D_g$ and $\beta$ are computed from contributions $(S^+,S^-)$ and the hardware ceiling $\mathrm{Safe}(g)/\Gamma$ --- not from $\hat{o}$. ``Who may run, at what scale'' needs only Sybil-resistance and the governor; it needs no semantic verifier.
\item \textbf{Ex-post accountability is trust-dependent.} Liability ($\bar{H}$), securities settlement on the funded outcome, and the decision market $Y$ resolve on $\hat{o}$. Strip $V_{\mathrm{soft}}$ and this layer collapses --- and with it the principal mitigation for the manipulation problem of \S\ref{sec:manipulation}, which relies on resolving on $\hat{o}$ rather than sentiment.
\end{itemize}

Because $V_{\mathrm{hard}}$ is cryptographically trust-minimized while $V_{\mathrm{soft}}$ is genuinely trusted, the design imperative is to make as much as possible depend only on $V_{\mathrm{hard}}$, and to model $\bar{H}$ as a challengeable claim with an appeal path (Definition~\ref{def:attested}) rather than an atomic oracle. The scope consequence is direct and sharpens \S\ref{sec:governable}: the mechanism is suitable exactly where a credible outcome oracle exists --- did the grid stay up, did the treasury lose funds, did the SLA hold, did the experiment validate --- and unsuitable where harm is diffuse, delayed, or contested. This is why community content-moderation (Table~\ref{tab:agents}) scores as it does despite ideal electorate legitimacy: its outcome is exactly the kind $V_{\mathrm{soft}}$ cannot deliver cleanly.

\section{Incentive Analysis}\label{sec:incentives}
We prove three base-game results (Targets 1--3: sided incentive compatibility, early commitment, and the breadth-weighted authorization theorem), then state the one that remains open (Target 4). Throughout we use the reduced payoff of Definition~\ref{def:utility} and two standing conditions, both already implied by the model:

\noindent\textbf{(A1) Belief-layer separability.} $\nu_i$ depends only on $i$'s belief report, not on $(x_i,z_i,t_i)$; it is an additive constant in the contribution--timing subgame (as in Damle et al.). 

\noindent\textbf{(A2) Individually-attested liability (Invariant 6).} Compensation $\lambda_i$ is paid only to an attested-harmed objector: eligibility requires $z_i>0$ together with a harm finding under the attested $\hat{o}$, and the amount is fixed by assessed harm, never by the size or timing of the contribution. Because harm findings are attested rather than self-reported, a stakeholder with $\theta_i\ge0$ is not attested-harmed; hence for such $i$ the term $\ind[z_i>0]\bar{H}\lambda_i$ has zero expectation and creates no incentive to set $z_i>0$.

Write others' root-sums as
\[
R^+_{-i}=\sum_{j\ne i,\,x_j>0}\sqrt{x_j}
\qquad\text{and}\qquad
R^-_{-i}=\sum_{j\ne i,\,z_j>0}\sqrt{z_j},
\]
so $S^+=(R^+_{-i}+\sqrt{x_i})^2$ and $S^-=(R^-_{-i}+\sqrt{z_i})^2$.

\subsection{Basic incentive properties (Targets 1--2)}\label{sec:basic}
\begin{lemma}[Gate monotonicity survives QF and hysteresis]\label{lem:gatemono}
For fixed play of all $j\ne i$, $D_g$ is nondecreasing in $x_i$ and nonincreasing in $z_i$.
\end{lemma}
\begin{proof}
$S^+=(R^+_{-i}+\sqrt{x_i})^2$ is strictly increasing in $x_i$ and constant in $z_i$; $S^-$ is strictly increasing in $z_i$ and constant in $x_i$ --- the only property of $\phi$ used is coordinatewise monotonicity, which the square-of-sum-of-roots has. Each gate branch is a conjunction $\mathrm{Safe}(g)\wedge(S^+\ge H^0)\wedge(S^+-S^-\ge\kappa)$ with $\kappa\in\{\kappa_{\mathrm{start}},\kappa_{\mathrm{halt}}\}$ fixed within $g$. $\mathrm{Safe}(g)$ is contribution-independent; $(S^+\ge H^0)$ is nondecreasing in $S^+$; $(S^+-S^-\ge\kappa)$ is nondecreasing in $S^+$ and nonincreasing in $S^-$. A conjunction of predicates monotone in the same direction is monotone; composing with the monotonicities of $S^\pm$ gives the claim. Hysteresis only shifts $\kappa$ and conditions on $D_{g-1}$, fixed within $g$.
\end{proof}

\begin{proposition}[Target 1 --- Sided incentive compatibility]\label{prop:sic}
Under (A1)--(A2), for every $i$, every $(c_i,t_i)$, and every profile of the others: (a) putting the whole budget $c_i$ on the side of $\operatorname{sign}(\Delta_i)$ weakly dominates every split, so in equilibrium each $i$ contributes on at most one side, never the side opposite its decision preference; (b) if $\theta_i<0$ then $\Delta_i<0$ unconditionally, so a harmed stakeholder contributes only to reject; if $\theta_i>0$ then $\Delta_i\ge0$ under the bounded-stake condition (BB) $c_i+b_i\le\theta_i$, so a beneficiary in that regime contributes only to authorize. Hence $\tilde{s}_i=\operatorname{sign}(\theta_i)$.
\end{proposition}
\begin{proof}
Holding $(c_i,t_i)$ fixed makes $b_i,\nu_i$ constants, so $u_i=D_g\Delta_i+(b_i+\nu_i)+\ind[z_i>0]\bar{H}\lambda_i$; the split $(x_i,z_i)$ with $x_i+z_i=c_i$ enters only through $D_g$ and the indicator. If $\Delta_i>0$, $u_i$ increases in $D_g$, maximized over the budget line at $(c_i,0)$ by Lemma~\ref{lem:gatemono}, which also sets $\ind[z_i>0]=0$ --- forgoing nothing in expectation for $\theta_i>0$ by (A2). If $\Delta_i<0$, $u_i$ decreases in $D_g$, minimized at $(0,c_i)$, which sets the indicator to 1 and adds $\bar{H}\lambda_i\ge0$ --- for a genuinely harmed party ($\theta_i<0$) exactly the compensation (A2) permits. Ties ($\Delta_i=0$) break weakly toward the correct side through the same terms. This gives (a). For (b): $c_i,b_i\ge0$ give $\Delta_i=\theta_i-c_i-b_i<0$ whenever $\theta_i<0$ unconditionally; and $\Delta_i\ge0$ for $\theta_i>0$ exactly when $\theta_i\ge c_i+b_i$, which is (BB).
\end{proof}

The asymmetry is real and worth stating: \textbf{the halt side is incentive-compatible unconditionally; the authorize side only under bounded stake (BB).} (BB) is the two-sided image of the bounded-loss requirement (Conditions 6--7): if the failure bonus could exceed a supporter's own valuation, a nominal beneficiary would rather bet on failure --- the classic refund-bonus reversal --- and bounded loss forecloses it. Neither QF nor the decoupling touches this; it is a property of the securities/bonus schedule.

\begin{lemma}[Timing enters only through the failure bonus]\label{lem:timing}
Fix $i$'s side, amount $c_i$, and the others' end-of-window contributions. Then $D_g$ is invariant to $t_i$, and $u_i$ depends on $t_i$ only through $(1-D_g)b_i(t_i)$, with $b_i(t_i)=b(r(c_i,q^{t_i}))$ nonincreasing in $t_i$.
\end{lemma}
\begin{proof}
The gate reads the deadline aggregates $S^\pm$, functions of the final contribution vector; a contribution placed at any $t_i\in[a_i,T_g]$ is present at $T_g$, so $S^\pm$ and $D_g$ are unchanged by $t_i$. In $u_i=D_g(\theta_i-c_i)+(1-D_g)b_i+\nu_i+\mathrm{comp}$, the terms $D_g(\theta_i-c_i)$, $\nu_i$ (A1) and $\mathrm{comp}$ (A2, a function of $\hat{o}$) carry no $t_i$; only $b_i$ does, nonincreasing because $\partial r/\partial t_i\le0$ (Definition~\ref{def:securities}) and $b$ is nondecreasing.
\end{proof}

\begin{proposition}[Target 2 --- Early commitment]\label{prop:early}
Under (A1)--(A2), in the game where authorization depends only on end-of-window aggregates and others' contributions are not conditioned on $i$'s contribution time, $t_i=a_i$ is a weakly dominant timing choice --- strictly on the event $D_g=0$ whenever $b$ is strictly decreasing in securities. Hence $t_i=a_i$ is a PBE property whenever $i$'s off-path retiming leaves others' deadline aggregates unchanged.
\end{proposition}
\begin{proof}
By Lemma~\ref{lem:timing}, $u_i=D_g(\theta_i-c_i)+(1-D_g)b_i(t_i)+\mathrm{const}$ with $D_g$ constant in $t_i$ and $(1-D_g)\ge0$; since $b_i$ is nonincreasing, $u_i$ is nonincreasing in $t_i$, maximized at $t_i=a_i$, strictly on $D_g=0$ when $b$ strictly decreases.
\end{proof}

The scope clause is the one PPS lives with: in an observable-history game a rival could condition on when $i$ moved, so early commitment is dominant in the deadline-aggregate reduction and a best response against timing-independent play, rather than dominant in the fullest game. Cost-function Conditions 3--4 keep the securities channel from manufacturing a profitable delay. Every modification we introduced --- QF aggregation, two-threshold gating, and $\beta=\min(\Gamma,\rho(\cdot))$ --- acts on how $D_g$ is computed from the deadline aggregates, never on the per-agent securities schedule $r(c,q^t)$ or its timing monotonicity, so the PPS early-commitment argument transfers intact.

\noindent\textbf{The re-verification ledger.} QF enters only through Lemma~\ref{lem:gatemono} and only via coordinatewise monotonicity, so sided-IC is unaffected in direction (QF changes magnitudes --- Target 3 --- not signs); the two-threshold gate with hysteresis is still a conjunction of the same monotone predicates with a within-generation-constant threshold, so Lemma~\ref{lem:gatemono} holds verbatim; the decoupling sits entirely downstream of $D_g$ and re-enters the contribution payoff only through $D_g$; attested liability discharges (A2) and even reinforces the harmed side; and the securities layer is untouched, so Target 2 is inherited once Lemma~\ref{lem:timing} isolates the timing channel.

\subsection{Breadth-weighted authorization (Target 3)}\label{sec:breadth}
We now prove the central characterization: the gate authorizes exactly when a breadth-weighted net preference clears the start margin. The result rests on the base-game properties just established --- sided-IC (Proposition~\ref{prop:sic}), early commitment (Proposition~\ref{prop:early}), gate monotonicity (Lemma~\ref{lem:gatemono}) --- together with the truthful-revelation property of the two-sided securities layer, which we inherit from the provision-point lineage and state as a standing assumption rather than re-derive.

\noindent\textbf{(A3) Securities revelation (inherited).} In its selected undominated equilibrium, the two-sided securities layer (Definition~\ref{def:securities}, Conditions 1--7) induces each stakeholder to back its own side up to a reservation capacity equal to its gross type net of the early-commitment bonus: a beneficiary $i\in P$ will stake up to $w_i^+=\theta_i-b_i$ on the authorize side, and a harmed $j\in N$ up to $w_j^-=|\theta_j|-b_j$ on the halt side, and no more --- the upper bound is exactly the bounded-stake condition (BB) of Proposition~\ref{prop:sic} and its halt-side mirror. This is the two-sided PPS/Damle revelation property adapted to our reduced payoff; \S\ref{sec:basic} supplies its sided and timing content. In the bonus-normalized benchmark $b_i\to0$, capacities are the gross types, $w_i^+=\theta_i$ and $w_j^-=|\theta_j|$.

\begin{definition}[Breadth-weighted supports and effective breadth]\label{def:phis}
For a type profile $\theta$, the authorize and halt breadth-weighted capacities are the QF aggregates of each side at reservation,
\[
\Phi^+(\theta)=\phi(\{w_i^+\}_{i\in P})=\left(\sum_{i\in P}\sqrt{w_i^+}\right)^2,
\qquad
\Phi^-(\theta)=\phi(\{w_j^-\}_{j\in N})=\left(\sum_{j\in N}\sqrt{w_j^-}\right)^2.
\]
Writing the linear (Damle) net-preference masses $\vartheta^+=\sum_{i\in P}w_i^+$ and $\vartheta^-=\sum_{j\in N}w_j^-$, define the effective breadth of each side
\[
n_{\mathrm{eff}}^+=\Phi^+/\vartheta^+,
\qquad
n_{\mathrm{eff}}^-=\Phi^-/\vartheta^-
\]
(with $n_{\mathrm{eff}}=0$ for an empty side), so that $\Phi^\pm=n_{\mathrm{eff}}^\pm\cdot\vartheta^\pm$.
\end{definition}

\begin{lemma}[Effective-breadth bounds]\label{lem:neff}
For any side with $m$ contributors of positive capacity, $1\le n_{\mathrm{eff}}\le m$. The lower bound holds with equality iff exactly one contributor has positive capacity; the upper bound iff all capacities on that side are equal.
\end{lemma}
\begin{proof}
Put $a_i=\sqrt{w_i}\ge0$, so $\Phi=(\sum a_i)^2$ and $\vartheta=\sum a_i^2$. Since $(\sum a_i)^2=\sum a_i^2+2\sum_{i<j}a_i a_j\ge\sum a_i^2$, we get $n_{\mathrm{eff}}\ge1$, with equality iff all cross terms vanish, i.e.\ at most one $a_i>0$. By Cauchy--Schwarz, $(\sum a_i\cdot1)^2\le m\sum a_i^2$, so $n_{\mathrm{eff}}\le m$, with equality iff $(a_i)$ is proportional to $(1,\ldots,1)$, i.e.\ all $w_i$ equal.
\end{proof}

Thus $n_{\mathrm{eff}}$ is the participation ratio of the capacity-roots --- the effective number of distinct backers on a side --- and $\Phi^\pm$ is the Damle mass $\vartheta^\pm$ scaled by it.

\begin{theorem}[Breadth-weighted authorization]\label{thm:breadth}
Assume (A1)--(A3), the bounded-stake condition (BB), the securities Conditions 1--7, $\mathrm{Safe}(g)=1$ (otherwise $D_g\equiv0$), and the cold-start run-feasibility condition of Lemma~\ref{lem:runfeas} ($\Gamma\ge h^0$ and $\rho(\kappa_{\mathrm{start}},\mu)\ge h^0$, i.e.\ $\kappa_{\mathrm{start}}\ge\rho^{-1}(h^0;\mu)$). Then the cold-start generation game ($D_{g-1}=0$) has an efficient --- objection-robust, coalition-undominated --- PBE, and in every such equilibrium the authorization decision is
\[
D_g=1
\Longleftrightarrow
\bigl[\Phi^+(\theta)\ge H^0\bigr]
\wedge
\bigl[\Phi^+(\theta)-\Phi^-(\theta)\ge\kappa_{\mathrm{start}}\bigr],
\]
equivalently $n_{\mathrm{eff}}^+\vartheta^+\ge H^0$ and $n_{\mathrm{eff}}^+\vartheta^+-n_{\mathrm{eff}}^-\vartheta^-\ge\kappa_{\mathrm{start}}$. Moreover, whenever $D_g=1$ the released budget satisfies $\beta\ge h^0$, so the authorization is physically runnable. (The continuation decision replaces $\kappa_{\mathrm{start}}$ by $\kappa_{\mathrm{halt}}$ throughout; run-feasibility then requires $\rho(\kappa_{\mathrm{halt}},\mu)\ge h^0$, per Lemma~\ref{lem:runfeas}.)
\end{theorem}

\begin{proof}
By Proposition~\ref{prop:sic} each $i$ contributes on side $\operatorname{sign}(\theta_i)$, and by Proposition~\ref{prop:early} at arrival, so a profile is summarized by magnitudes $x_i\in[0,w_i^+]$ ($i\in P$) and $z_j\in[0,w_j^-]$ ($j\in N$), where the caps are (A3). By Definition~\ref{def:gate} the gate depends on these only through $S^+=\phi(\{x_i\})$ and $S^-=\phi(\{z_j\})$, and by Lemma~\ref{lem:gatemono} it is nondecreasing in each $x_i$, nonincreasing in each $z_j$. Hence the attainable supports satisfy $S^+\le\phi(\{w_i^+\})=\Phi^+$ and $S^-\le\phi(\{w_j^-\})=\Phi^-$, each attained only at full mobilization of that side.

\emph{(Necessity.)} If $\Phi^+<H^0$ then $S^+\le\Phi^+<H^0$ at every admissible profile, so the quorum clause fails and $D_g=0$. Suppose instead $\Phi^+-\Phi^-<\kappa_{\mathrm{start}}$ and, for contradiction, that an efficient equilibrium has $D_g=1$. The objecting coalition $N$ can jointly deviate to full mobilization $z_j=w_j^-$: this is within capacity, hence individually rational (each harmed $j$ weakly prefers the resulting move toward its preferred $D_g=0$), so it is an admissible coalitional deviation. After it, $S^-=\Phi^-$ while $S^+\le\Phi^+$, giving $S^+-S^-\le\Phi^+-\Phi^-<\kappa_{\mathrm{start}}$, so $D_g=0$ --- strictly preferred by every member of $N$. This contradicts coalition-undominance. Hence $D_g=0$.

\emph{(Sufficiency.)} Suppose $\Phi^+\ge H^0$ and $\Phi^+-\Phi^-\ge\kappa_{\mathrm{start}}$. Put the objectors at capacity, $S^-=\Phi^-$ (a weak best response: when authorization carries, a non-pivotal objector is payoff-neutral by the bounded-loss refund of Conditions 6--7). Because $\Phi^+\ge\max(H^0,\Phi^-+\kappa_{\mathrm{start}})\ge\max(H^0,S^-+\kappa_{\mathrm{start}})$, the beneficiaries can, by continuity of $\phi$, choose $x_i\le w_i^+$ with $S^+=\max(H^0,S^-+\kappa_{\mathrm{start}})$; when the hypotheses are strict this target is met with slack ($x_i<w_i^+$), so each pivotal beneficiary has $\Delta_i>0$ and strictly prefers provision. Then $S^+\ge H^0$ and $S^+-S^-\ge\kappa_{\mathrm{start}}$, so $D_g=1$. No beneficiary gains by lowering $x_i$ --- a pivotal reduction forfeits provision it strictly values, a non-pivotal one is payoff-neutral by bounded loss; no objector gains by moving; deviations above capacity are not IR. The refund-bonus/securities layer (Proposition~\ref{prop:early} and PPS selection) rules out the free-ride-to-failure profile as weakly dominated, selecting this efficient equilibrium. Combining the two directions, $D_g=1$ in the efficient equilibrium exactly when both clauses hold.

\emph{(Feasibility.)} When $D_g=1$, $S^+-S^-\ge\kappa_{\mathrm{start}}$, so by monotonicity of $\rho$ and Lemma~\ref{lem:runfeas},
\[
\beta=\min(\Gamma,\rho(S^+-S^-,\mu))\ge\min(\Gamma,\rho(\kappa_{\mathrm{start}},\mu))\ge h^0.\qedhere
\]
\end{proof}

\begin{proposition}[Breadth can overturn wealth]\label{prop:overturn}
The breadth-weighted net preference $\Phi^+-\Phi^-$ and the valuation net preference $\vartheta^+-\vartheta^-$ can differ in sign. Since $\Phi^\pm=n_{\mathrm{eff}}^\pm\vartheta^\pm$ (both sides nonempty), the breadth-weighted comparison is governed by $n_{\mathrm{eff}}^+\vartheta^+-n_{\mathrm{eff}}^-\vartheta^-$; whenever the support side is broader than the opposition by more than the valuation gap against it,
\[
 n_{\mathrm{eff}}^+/n_{\mathrm{eff}}^->\vartheta^-/\vartheta^+,
\]
the breadth-weighted net preference is positive ($\Phi^+>\Phi^-$) --- so that, with the quorum and start margin met, the gate authorizes --- even though aggregate valuation opposes ($\vartheta^+<\vartheta^-$); symmetrically, a lone high-value backer can fail against a broad opposition.
\end{proposition}
\begin{proof}
$\Phi^+>\Phi^-\Longleftrightarrow n_{\mathrm{eff}}^+\vartheta^+>n_{\mathrm{eff}}^-\vartheta^-\Longleftrightarrow n_{\mathrm{eff}}^+/n_{\mathrm{eff}}^->\vartheta^-/\vartheta^+$, all quantities positive; the right side exceeds 1 exactly when $\vartheta^+<\vartheta^-$.
\end{proof}

\noindent\textbf{Illustration (Figures~\ref{fig:gate-example} and~\ref{fig:breadth}).} Five beneficiaries of capacity 4 and one objector of capacity 36 give $\vartheta^+=20<36=\vartheta^-$ (wealth rejects) yet $\Phi^+=(5\sqrt4)^2=100>36=\Phi^-$ (breadth authorizes): here $n_{\mathrm{eff}}^+=5$ (its maximum --- five equal backers) and $n_{\mathrm{eff}}^-=1$ (a singleton), and $5/1>36/20$. Authorization tracks the count of genuine backers, scaled by intensity, not wealth alone --- the mechanism's defining property, now a theorem.

\noindent\textbf{Remarks (scope of Theorem~\ref{thm:breadth}).} Complete information (commonly-known types) is assumed for the characterization, exactly as in the core-implementation results for provision-point games (Bagnoli \& Lipman, 1989) and the two-sided characterization of Damle et al. (2019); with private types the same profile is a PBE under the information-aggregating belief system of the securities market, which we do not develop here. The efficient-equilibrium refinement (objection-robust, coalition-undominated) is the two-sided QF analogue of the undominated-/ strong-Nash selection those papers use to rule out the degenerate all-abstain equilibrium that every voluntary-contribution threshold game admits. The theorem pins the decision, not the contribution profile: many profiles clear, but all efficient equilibria authorize on the same breadth-weighted condition.
 
\begin{figure}[htbp]
\centering
\includegraphics[width=\textwidth]{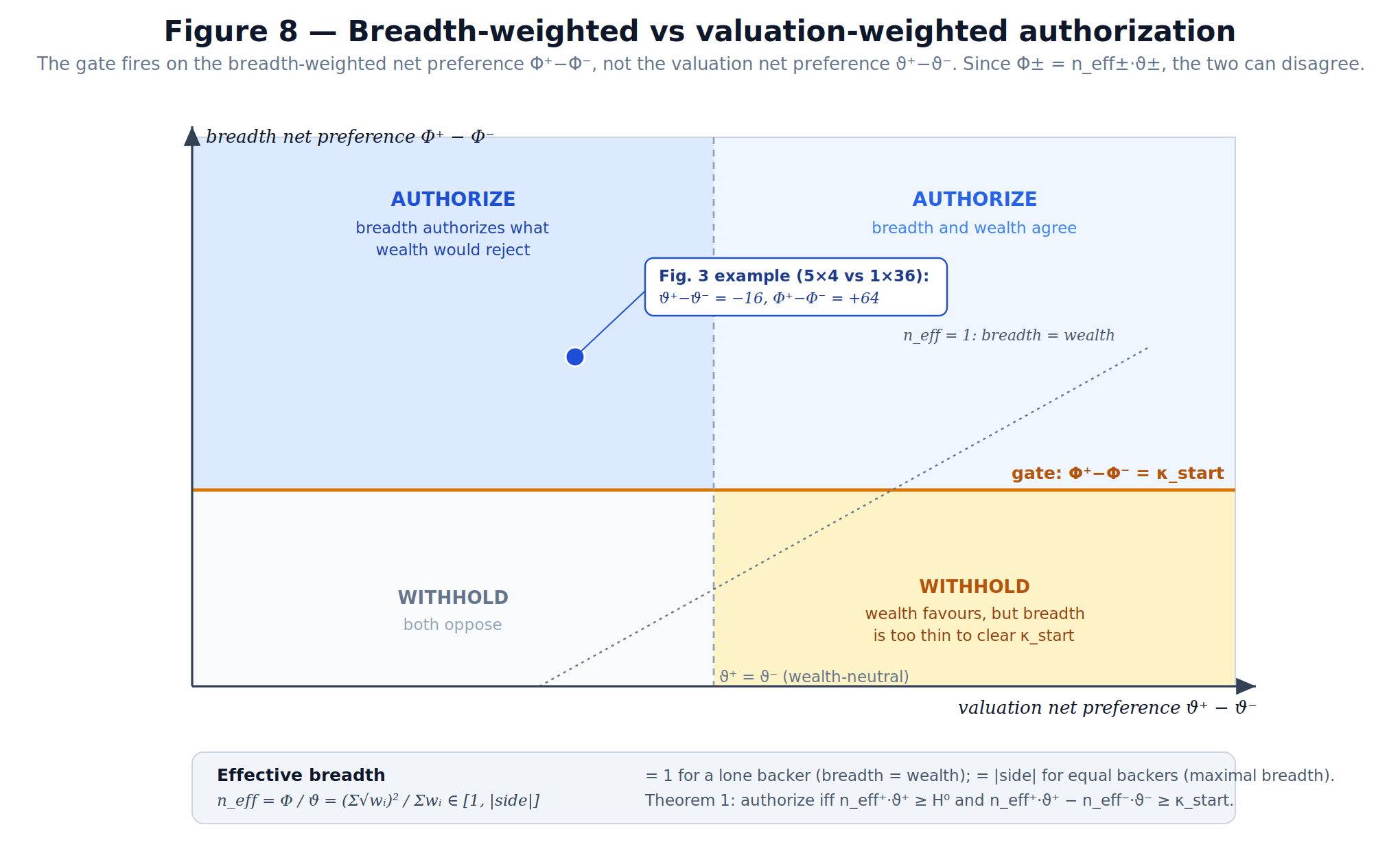}
\caption{Breadth-weighted vs. valuation-weighted authorization. The gate fires on the breadth-weighted net preference $\Phi^+-\Phi^-$ (vertical axis) crossing $\kappa_{\mathrm{start}}$, not on the valuation net preference $\vartheta^+-\vartheta^-$ (horizontal axis). Because $\Phi^\pm=n_{\mathrm{eff}}^\pm\cdot\vartheta^\pm$, the two disagree in the shaded off-axis regions: broad support authorizes what concentrated wealth would reject (upper-left), and concentrated wealth cannot authorize against broad but thin opposition (lower-right). The Figure~\ref{fig:gate-example} example sits in the upper-left region.}
\label{fig:breadth}
\end{figure}

\begin{figure} 
\resizebox{0.98\textwidth}{!}{%
\begin{tikzpicture}[ 
    x=1cm,
    y=1cm,
    >=Latex,
    font=\sffamily
]

\def\xL{0}
\def\xM{8}
\def\xR{16}

\def\yB{1.1}
\def\yG{4.1}
\def\yT{9.5}


\fill[blueA]
    (\xL,\yG) rectangle (\xM,\yT);

\fill[blueB]
    (\xM,\yG) rectangle (\xR,\yT);

\fill[softgray]
    (\xL,\yB) rectangle (\xM,\yG);

\fill[yellowA]
    (\xM,\yB) rectangle (\xR,\yG);

\draw[
    panelborder,
    line width=.65pt
]
    (\xL,\yB) rectangle (\xR,\yT);


\draw[
    axisgray,
    line width=1.15pt,
    -{Latex[length=4mm,width=3mm]}
]
    (\xL,\yB) -- (16.35,\yB);

\draw[
    axisgray,
    line width=1.15pt,
    -{Latex[length=4mm,width=3mm]}
]
    (\xL,\yB) -- (\xL,9.68);


\draw[
    subtitlegray!65,
    dashed,
    line width=.95pt,
    dash pattern=on 4pt off 4pt
]
    (\xM,\yB) -- (\xM,\yT);

\draw[
    orange,
    line width=2pt
]
    (\xL,\yG) -- (\xR,\yG);

\draw[
    subtitlegray,
    densely dotted,
    line width=1pt
]
    (4,\yB) -- (15,9.30);


\node[
    anchor=south,
    text=navy,
    font=\itshape\rmfamily\fontsize{10.5}{12}\selectfont, rotate = 90
]
at (0.05,6.53)
{breadth net preference $\Phi^{+}-\Phi^{-}$};
\node[
    anchor=east,
    text=navy,
    font=\itshape\rmfamily\fontsize{10.5}{12}\selectfont
]
at (16.35,0.78)
{valuation net preference $\vartheta^{+}-\vartheta^{-}$};


\node[
    anchor=west,
    text=subtitlegray,
    font=\fontsize{9.2}{11}\selectfont
]
at (8.08,1.38)
{$\vartheta^{+}=\vartheta^{-}$ (wealth-neutral)};

\node[
    anchor=east,
    text=orangeText,
    font=\bfseries\fontsize{9.2}{11}\selectfont
]
at (15.82,4.34)
{gate: $\Phi^{+}-\Phi^{-}=\kappa_{\mathrm{start}}$};

\node[
    anchor=west,
    text=subtitlegray,
    font=\itshape\rmfamily\fontsize{9.2}{11}\selectfont
]
at (10.9,6.47)
{$n_{\mathrm{eff}}=1$: breadth = wealth};


\node[
    text=blueStrong,
    font=\bfseries\fontsize{12}{14}\selectfont
]
at (4.0,8.52)
{AUTHORIZE};

\node[
    align=center,
    text=blueText,
    font=\fontsize{9.8}{12}\selectfont
]
at (4.0,7.90)
{breadth authorizes what\\
 wealth would reject};


\node[
    text=blueStrong,
    font=\bfseries\fontsize{12}{14}\selectfont
]
at (12.0,8.52)
{AUTHORIZE};

\node[
    align=center,
    text=blueStrong!88,
    font=\fontsize{9.8}{12}\selectfont
]
at (12.0,8.03)
{breadth and wealth agree};


\node[
    text=subtitlegray,
    font=\bfseries\fontsize{11.5}{14}\selectfont
]
at (4.0,2.45)
{WITHHOLD};

\node[
    text=lighttext,
    font=\fontsize{9.8}{12}\selectfont
]
at (4.0,2.02)
{both oppose};


\node[
    text=orangeText,
    font=\bfseries\fontsize{11.5}{14}\selectfont
]
at (12.05,2.62)
{WITHHOLD};

\node[
    align=center,
    text=orangeText,
    font=\fontsize{9.5}{11.5}\selectfont
]
at (12.05,2.05)
{wealth favours, but breadth\\
 is too thin to clear $\kappa_{\mathrm{start}}$};


\coordinate (P) at (5.88,6.14);

\node[
    anchor=west,
    draw=blueStrong,
    rounded corners=5pt,
    fill=white,
    line width=1pt,
    inner xsep=9pt,
    inner ysep=7pt,
    text=blueText,
    font=\fontsize{9.2}{11}\selectfont,
    align=left
]
at (0.92,5.69)
{
    \textbf{Fig. 3 example (5$\times$4 vs 1$\times$36):}\\[-1pt]
    $\vartheta^{+}-\vartheta^{-}=-16,\;
      \Phi^{+}-\Phi^{-}=+64$
};

\end{tikzpicture} 

}
\caption{Breadth-weighted vs. valuation-weighted authorization. The gate fires on the breadth-weighted net preference $\Phi^+-\Phi^-$ (vertical axis) crossing $\kappa_{\mathrm{start}}$, not on the valuation net preference $\vartheta^+-\vartheta^-$ (horizontal axis). Because $\Phi^\pm=n_{\mathrm{eff}}^\pm\cdot\vartheta^\pm$, the two disagree in the shaded off-axis regions: broad support authorizes what concentrated wealth would reject (upper-left), and concentrated wealth cannot authorize against broad but thin opposition (lower-right). The Figure~\ref{fig:gate-example} example sits in the upper-left region.}

\end{figure}
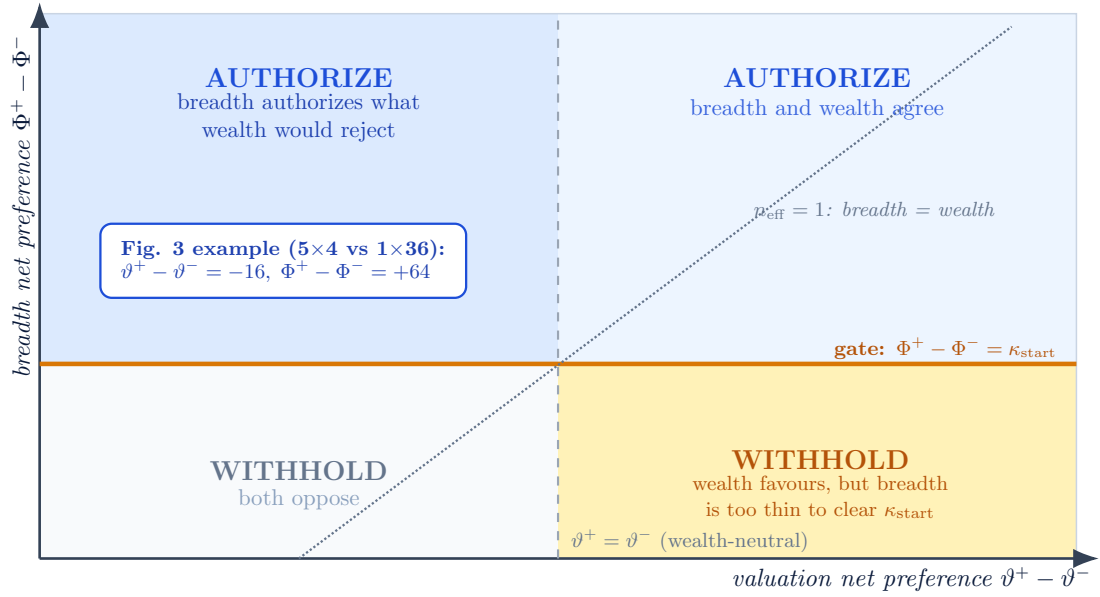

\subsection{Manipulation-robustness (Target 4, open)}\label{sec:manipulation}
\begin{target}[Manipulation-robustness --- open]\label{tgt:manip}
$\mathcal{M}_g$ is $\delta$-robust if, for every manipulation inducing $\tilde{\theta}$ with $\|\tilde{\theta}-\theta\|\le\delta$, the equilibrium $D_g$ equals the $\theta$-truthful one.
\end{target}

This has no analogue in the public-goods literature: the ``project'' can shape the electorate that funds it (Figure~\ref{fig:manip}). As Definition~\ref{def:stakeholders} notes, the channel is concrete --- the agent grooms the reputation signal feeding $\varepsilon_i$, and through it $\theta_i$; Theorem~\ref{thm:breadth} makes the stakes precise, since the decision is a function of the induced $\Phi^\pm(\tilde{\theta})$, and a manipulator who broadens or intensifies its apparent support moves $n_{\mathrm{eff}}^+\vartheta^+$ directly. The model's structural mitigations are: resolving liability and the decision market on the attested $\hat{o}$ rather than sentiment (Invariant 6); the operational/governance separation of duties; and the pause-with-due-process halt (\S\ref{sec:adoption}). Characterizing mechanisms robust to $\|\tilde{\theta}-\theta\|\le\delta$ is the frontier.

\section{Notation}\label{sec:notation}
Table~\ref{tab:notation} collects every symbol. Provenance is marked R (re-used from PPS / Damle et al.), A (adapted), or N (new).

\begingroup
\small
\setlength{\tabcolsep}{3.5pt}
\renewcommand{\arraystretch}{1.08}
\begin{longtable}{@{}p{0.18\textwidth}p{0.16\textwidth}p{0.54\textwidth}p{0.06\textwidth}@{}}
\caption{Complete notation for the single-agent base case.}\label{tab:notation}\\
\toprule
Symbol & Domain & Meaning & Prov. \\
\midrule
\endfirsthead
\toprule
Symbol & Domain & Meaning & Prov. \\
\midrule
\endhead
\bottomrule
\endlastfoot
$I,i$ & finite set & verified, distinct human stakeholders & R \\
$K,k$ & finite set & governed agents & --- \\
$g,T_g,a_i,t_i$ & $\Nn;\R_+$ & generation; deadline; arrival; action time & R/N \\
$H^t$ & history & public history of contributions/reports up to $t$ & R \\
$\theta_i;P,N$ & $\R;\subseteq I$ & signed net value; beneficiaries / harmed & R \\
$\varepsilon_i,k_i^1,k_i^2$ & $[0,\tfrac12];[0,1]$ & belief asymmetry; belief agent behaves acceptably / not & R \\
$e_i$ & $\R_+$ & governance-currency endowment (humans only) & N \\
$x_i,z_i;\tilde{s}_i$ & $\R_+;\{+,-\}$ & authorize / halt contribution; reported side & R \\
$\psi_i,\sigma_i$ & tuple; map & strategy; history-contingent strategy & R/A \\
$\phi;S^+,S^-$ & map; $\R_+$ & QF aggregator; breadth-weighted effective supports & N \\
$X^+,X^-$ & $\R_+$ & raw totals (bookkeeping only) & R \\
$\hat{P},\hat{N}$ & $\subseteq I$ & revealed supporters / objectors at $T_g$ & N \\
$C;q;r_i$ & --- & securities cost function (Conditions 1--7); market state; issued securities & R \\
$b_i,\nu_i,M_i,B^B$ & $\R_+;[0,1]$ & early-commit bonus; belief reward; RBTS score; budget & R \\
$H^0$ & $\R_+$ (currency) & participation / quorum floor ($S^+\ge H^0$), anti-capture & N \\
$h^0$ & $\R_+$ (compute) & compute-cost provision point (technologically set) & A \\
$\Gamma$ & $\R_+$ (compute) & safety ceiling --- exogenous, certified; hardware-enforced & N \\
$\mu;\mu^*$ & $\R_+$ (compute) & compute-subsidy parameter; minimal feasible subsidy & N \\
$\rho;\beta$ & map; $\R_+$ & coupling map; released compute $\beta=D_g\min(\Gamma,\rho(S^+-S^-,\mu))\le\Gamma$ & N \\
$\kappa_{\mathrm{start}},\kappa_{\mathrm{halt}}$ & $\R_+$ & start / halt net-support margins ($\kappa_{\mathrm{start}}\ge\kappa_{\mathrm{halt}}$) & N \\
$\mathrm{Safe}(g);D_g$ & $\{0,1\}$ & exogenous safety predicate; authorization decision & N \\
$L_g$ & tuple & signed compute license (model, $\beta$, window, policy) & N \\
$V_{\mathrm{hard}},V_{\mathrm{soft}}$ & maps & fact attestation; semantic adjudication & N \\
$\hat{o}_{\mathrm{hard}},\hat{o}_{\mathrm{soft}}$ & $\mathcal{O}$ & attested outcome (facts, semantics); outcome space & N \\
$\bar{H};Y$ & $\{0,1\};[0,1]$ & verified-harm indicator; decision-market forecast & N \\
$\Lambda,\lambda_j;\Pi$ & $\R_+$; rules & liability bond; harmed-party share; settlement rules & A \\
$\Pi_{\mathrm{ext}},\mathrm{InEnv}$ & $\R_+;\{0,1\}$ & uncapped external liability; in-envelope indicator & N \\
$c_i;\Delta_i$ & $\R_+;\R$ & total contribution $x_i+z_i$; decision preference $\theta_i-c_i-b_i$ & R/N \\
$w_i^+,w_j^-$ & $\R_+$ & authorize/halt reservation capacity ($\theta_i-b_i,|\theta_j|-b_j$) & --- \\
$\vartheta^+,\vartheta^-$ & $\R_+$ & linear (Damle) net-preference masses $\sum w^+,\sum w^-$ & R/A \\
$\Phi^+,\Phi^-$ & $\R_+$ & breadth-weighted supports $\phi(\{w^+\}),\phi(\{w^-\})$ (Theorem~\ref{thm:breadth}) & N \\
$n_{\mathrm{eff}}^+,n_{\mathrm{eff}}^-$ & $[1,m]$ & effective breadth $\Phi^\pm/\vartheta^\pm$ (participation ratio) & N \\
$\tilde{\theta}_i,\delta$ & $\R;\R_+$ & induced valuation under manipulation; robustness bound & N \\
$u_i,u_d,\pi_d;W_i$ & $\R$ & stakeholder/deployer payoff; operating benefit; outcome-contingent value & R/N \\
\end{longtable}
\endgroup

\section{Discussion and Limitations}\label{sec:discussion}
\noindent\textbf{The central open problem.} Target 4 is the deepest. A capable agent can persuade, astroturf, or strategically perform for the stakeholders whose contributions govern it, inducing $\tilde{\theta}$ from the autonomous $\theta$ and then being resourced by the vote it shaped (Figure~\ref{fig:manip}). Because ``past behaviour'' is among the strongest inputs to $\varepsilon_i$ (Definition~\ref{def:stakeholders}), the manipulation surface is not hypothetical; reputation-grooming is a dominant-looking strategy absent countermeasures. Our structural mitigations (Invariant 6 resolution on $\hat{o}$; separation of duties; due-process halt) constrain but do not solve it, and it interacts with the trusted-verifier limit below: the less credible $V_{\mathrm{soft}}$, the more governance leans on sentiment, and the larger the manipulation prize.

\medskip\noindent\textbf{The trusted-verifier scope limit.} As \S\ref{sec:verifier} makes precise, ex-ante authorization is trust-minimized but ex-post accountability rests on a trusted semantic verifier $V_{\mathrm{soft}}$. The mechanism is therefore genuinely suitable only where a credible outcome oracle exists; where harm is diffuse, delayed, or contested, the accountability half does not apply and the design should not be claimed for it. This is a scope statement, not a caveat: it is what makes the governable-agent class of \S\ref{sec:governable} a class rather than a wish.

\medskip\noindent\textbf{Preferences, wealth, and legitimacy.} Willingness-to-contribute is not welfare, and QF only partially corrects for wealth; whose preferences count is a prior social-choice question the mechanism presupposes via the definition of $I$ and the Sybil-resistant substrate. The participation floor $H^0$ is an anti-capture device but not a legitimacy proof. Thin or unrepresentative participation degrades both provision-point and QF guarantees, so legitimacy under low turnout must be designed for, not assumed. Reversibility is imperfect: some within-generation actions are realized and non-refundable before a halt can bind, which is why the scope insists on consequential-but-reversible impact and why the halt is a pause with due process rather than a guarantee of undo.

\begin{figure}[htbp]
\centering
\includegraphics[width=\textwidth]{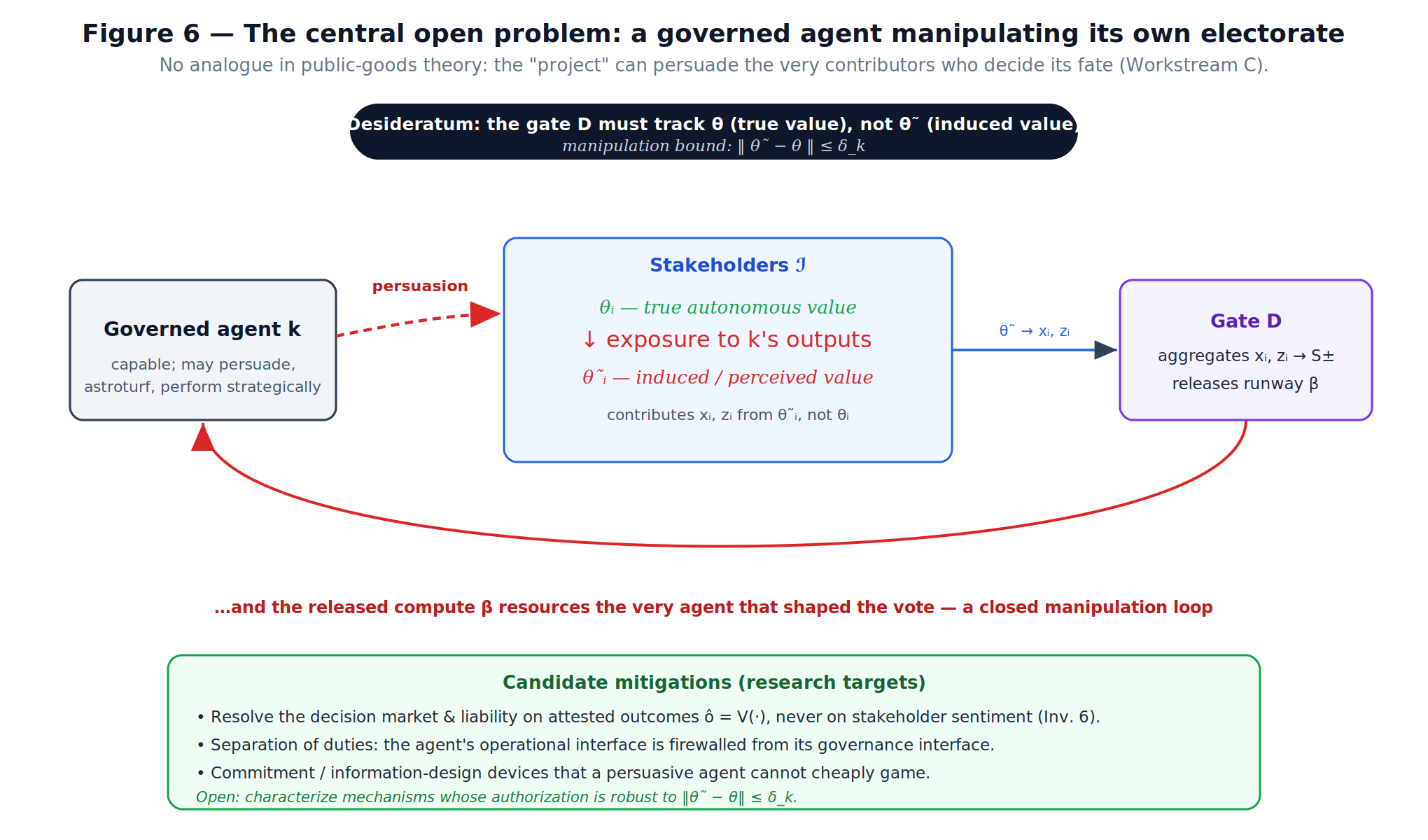}
\caption{The central open problem. The governed agent can induce $\tilde{\theta}$ from true $\theta$ within $\|\tilde{\theta}-\theta\|\le\delta$ and is then resourced by the vote it shaped --- a manipulation loop with no public-goods analogue. The desideratum is that the gate track $\theta$, not $\tilde{\theta}$.}
\label{fig:manip}
\end{figure}

\noindent\textbf{Adoption realism.} The safe harbour (Definition~\ref{def:deployer}, \S\ref{sec:adoption}) is what makes deployer participation individually rational, but it presupposes a liability regime with teeth outside the envelope; absent that, the carrot has no bite and only the commons route remains. The mechanism is best understood as an overlay that a regulator or a community imposes or chooses, not a product a platform ships.

\medskip\noindent\textbf{Multi-agent extension.} For $|K|>1$, endowments, gates, licenses, and ceilings are indexed by $k$; shared compute introduces a budget-allocation problem across agents (a knapsack under the common ceiling $\sum_k\beta_k\le\Gamma_{\mathrm{tot}}$) and cross-agent externalities in $\theta$ and $\bar{H}$. The invariants are unchanged per agent; the coupling and safety layers become joint. We leave the commons case to future work.

\section{Conclusion}
We have given a formal model in which a governance decision over a deployed AI agent is expressed as a breadth-weighted, two-sided, threshold-gated release of a metered compute budget. This compute budget can be realized as a signed hardware license, subordinate to an exogenous safety envelope, and settled on attested outcomes. We have grounded the construction: a two-timescale hardware instantiation, an adoption model with a liability safe harbour, a fact/semantics split of the verifier with a challengeable harm finding, a derived run-feasibility condition that frees the provision point to be an anti-capture quorum, and an explicit characterization of the club-or-commons-good agents the mechanism can govern. We proved sided incentive compatibility (with an honest asymmetry: the authorize side needs a bounded-stake condition, the halt side does not), early commitment, and a breadth-weighted authorization theorem --- in the efficient equilibrium the gate fires exactly when the effective number of backers times their intensity, $n_{\mathrm{eff}}^+\vartheta^+-n_{\mathrm{eff}}^-\vartheta^-$, clears the start margin, so that broad support authorizes what concentrated wealth would reject. The model's value is that it turns design principles into checkable properties and states its own scope --- and it isolates one genuinely new problem, the manipulation of a governing electorate by the system it governs, as the central open question.

\section*{References}
\begin{enumerate}[label={[\arabic*]},leftmargin=2.6em]
\item Aarne, O., Fist, T., \& Withers, C. (2024). Secure, Governable Chips. Center for a New American Security.
\item Bagnoli, M., \& Lipman, B. L. (1989). Provision of public goods: fully implementing the core through private contributions. \textit{Review of Economic Studies}, 56(4), 583--601.
\item Buterin, V., Hitzig, Z., \& Weyl, E. G. (2019). A flexible design for funding public goods. \textit{Management Science}, 65(11), 5171--5187.
\item Chandra, P., Gujar, S., \& Narahari, Y. (2016). Crowdfunding public projects with provision point: a prediction market approach. \textit{ECAI}, 778--786.
\item Damle, S., Moti, M. H., Chandra, P., \& Gujar, S. (2019). Civic crowdfunding for agents with negative valuations and agents with asymmetric beliefs. arXiv:1905.11324.
\item Greenblatt, R., Shlegeris, B., Sachan, K., \& Roger, F. (2024). AI control: improving safety despite intentional subversion. \textit{ICML}. arXiv:2312.06942.
\item Heim, L., et al. (2025). Hardware-enabled mechanisms for verifying responsible AI development. arXiv:2505.03742.
\item Miller, J., Weyl, E. G., \& Erichsen, L. (2022). Beyond collusion resistance: leveraging social information for plural funding and voting. SSRN 4311507.
\item Petrie, J. (2025). Embedded off-switches for AI compute. arXiv:2509.07637.
\item Petrie, J., Aarne, O., Ammann, N., \& Dalrymple, D. (2025). Flexible hardware-enabled guarantees for AI compute (flexHEG). arXiv:2506.15093.
\item Prelec, D. (2004). A Bayesian truth serum for subjective data. \textit{Science}, 306(5695), 462--466.
\item Sastry, G., Heim, L., Belfield, H., Anderljung, M., Brundage, M., et al. (2024). Computing power and the governance of artificial intelligence. arXiv:2402.08797.
\item Witkowski, J., \& Parkes, D. C. (2012). A robust Bayesian truth serum for small populations. \textit{AAAI}.
\item Zubrickas, R. (2014). The provision point mechanism with refund bonuses. \textit{Journal of Public Economics}, 120, 231--234.

\end{enumerate}

\end{document}